\documentclass[a4paper,12pt]{article}

\usepackage{amsthm,amsmath,amssymb,amsfonts,xcolor,mathtools}
\usepackage[normalem]{ulem}

\def \ring {{\cal R}}

\def \ord {\, \mbox{ord}}

\def \Ord {\, \mbox{Ord}}

\def \E{\rm E}

\newcommand{\vt}[2]{\rm _{#1}V_{#2}}
\newcommand{\vpt}[2]{\rm _{#1}V'_{#2}}
\newcommand{\tot}[2]{\rm _{#1}T_{#2}}
\newcommand{\bt}[2]{\rm _{#1}B_{#2}}

\def\raph{{\, \xmapsto{\ \phi\ }\,}}

\def\C{{\mathbb C}}
\def\Z{{\mathbb Z}}
\def\N{{\mathbb N}}
\def\A{{\mathfrak{A}}}

\def\hu{{\hat{u}}}

\def\cL{{\cal L}}
\def\cM{{\cal M}}
\def\cR{{\cal R}}

\def\cC{{\cal C}}

\def\cT{{\cal T}}

\def\D{{\cal D}}
\def \i{{\rm i}}
\def \h{\psi}

\def\s3{\sqrt{3}}

\newtheorem{Def}{Definition}
\newtheorem{The}{Theorem}
\newtheorem{Pro}{Proposition}

\newtheorem{Ex}{Example}

\title{Integrability of Nonabelian Differential-Difference Equations: the Symmetry Approach }
\author{Vladimir Novikov$^{a}$, Jing Ping Wang$^{b}$\\
{\small$^{a}$ School of Mathematics, Loughborough University, Loughborough, UK,}\\
{\small$^b$ School of Mathematics, Statistics and Actuarial Science, University of Kent, Canterbury, UK.}}
\date{}

\begin{document}

\maketitle

\begin{abstract}
We propose a novel approach to tackle integrability problem for evolutionary differential-difference equations (D$\Delta$Es) on free associative algebras, also referred to as nonabelian D$\Delta$Es. This approach enables us to derive necessary integrability conditions, determine the integrability of a given equation, and make progress in the classification of integrable nonabelian D$\Delta$Es. This work involves establishing symbolic representations for the nonabelian difference algebra, difference operators, and formal series, as well as introducing a novel quasi-local extension for the algebra of formal series within the context of symbolic representations. Applying this formalism, we solve the classification problem of integrable skew-symmetric quasi-linear nonabelian equations of orders $(-1,1)$, $(-2,2)$, and $(-3,3)$, consequently revealing some new equations in the process.
\end{abstract}

\section{Introduction}
In the recent paper \cite{MNWDiff}, Mikhailov, Novikov and Wang presented the perturbative symmetry approach for evolutionary differential-difference equations (D$\Delta$Es), which is based on establishing a symbolic representation for the difference polynomial ring and introducing its quasi-local extension. This approach enabled the derivation of necessary integrability conditions for  evolutionary D$\Delta$Es of arbitrary order and consequently allowed to solve some classification problems of integrable equations for any given order.
In this work, the existence of higher symmetries of a D$\Delta$Es, or more precisely, an infinite hierarchy of higher symmetries, is taken as the criterion for its integrability.

This paper can be considered as an extension of this research, transitioning from the commutative (abelian) case to the noncommutative (nonabelian) case.
Our focus lies in exploring the integrability conditions for evolutionary D$\Delta$Es when field variables take values in an associative, but noncommutative algebra.
Typical examples of associative algebras include matrix or operator algebras. Matrix formulations of integrable equations trace back to the very beginning of integrable systems theory. The first matrix integrable equation, known as the matrix KdV equation, was introduced by P. Lax in \cite{Lax}. One striking feature observed early on in integrable matrix equations is that their algebraic structures  such as symmetries, conservation laws, and Lax representations are independent of individual matrix entries and  the sizes of matrices. Instead, they are expressed  as polynomials or formal series in field variables. Consequently, field variables can be treated as elements of an associative, noncommutative algebra. When considering nonabelian D$\Delta$Es, we exclusively employ the intrinsic algebraic operations of multiplication, addition and
scalar multiplication, irrespective of any specific concrete realisation.

The systematic classification of integrable partial differential equations (PDEs) with field variables in an associative algebra was initiated by
Olver and Sokolov \cite{OS, OS2}. They provided a comprehensive list of integrable scalar evolution PDEs of this type, characterised by higher order symmetries. The completeness of this classification was demonstrated in \cite{OW}. In 2000, Mikhailov and Sokolov \cite{MS} successfully developed the theory of integrable ordinary differential equations on associative algebras. They solved a number of classification problems using the existence of hierarchies of first integrals and/or symmetries as a criterion for integrability.
In a recent contribution \cite{AS}, Adler and Sokolov extended classifications to nonabelian generalisations of integrable systems of nonlinear Schr{\"o}dinger and Boussinesq types. Their approach relied on both the existence of higher conservation laws and higher symmetries, which are determined by their corresponding abelian integrable systems.
Additionally, the classification of non-abelian Painlev{\'e} type systems was addressed in \cite{BS23}, introducing non-abelian constants into the equations.

Nonabelian integrable D$\Delta$Es have historically emerged as direct matrix generalisations and discretisations of matrix integrable PDEs. The matrix Volterra chain was initially introduced in \cite{MW}. Further examples of  integrable D$\Delta$Es on associative algebras, including the nonabelian Bogoyavlensky equation, were obtained in \cite{BG}. The matrix Toda lattice initially originated from a discrete version of the principal chiral field model \cite{bmrl80, mik81} and has recently appeared in the study of Matrix-valued Hermite polynomials \cite{b}. The algebraic structures underlying these equations, such as recursion operators and Hamiltonian structures, have been extensively investigated in \cite{CW,CasatiWang}.

The symmetry approach based on a concept of formal recursion operator has been formulated and developed in works of Shabat and co-authors
(see for example review papers \cite{SokShab,MikShabYam, MikShabSok}). It has become one of powerful tools for addressing the integrability problem.
Formal recursion operator carries information
about integrability, remaining robust even in the presence of gaps within the infinite hierarchy of symmetries or conservation laws.
There have been numerous exhaustive classification results of abelian integrable PDEs including nonlocal and/or non-evolutionary equations such as the Benjamin–Ono equation and the Camassa–Holm equation. The symmetry approach has recently been extended to D$\Delta$Es in \cite{MNWDiff} following the introduction of symbolic representation of the difference polynomial ring and its proper ring extension. However, the exploration of this approach in the realm of nonabelian integrable systems remains uncharted territory.

In this paper, we develop the perturbative symmetry approach, built upon the recent work in \cite{MNWDiff}, to investigate the integrability problem concerning evolutionary D$\Delta$Es
\begin{equation}
\label{begeq}
u_t=F(u_p,\ldots,u_q), \quad p\leq q\in \mathbb{Z}
\end{equation}
of order $(p, q)$, where $u=u(n,t)$ is a function of a discrete variable $n\in\Z$ and a continuous variable $t\in\C$, taking values in an associative algebra. We adopt standard notation:
$$
u_k=S^ku(n,t)=u(n+k,t),\quad u_t=\partial_t u,
$$
and $S$ is the shift operator. One main challenge from the abelian to nonabelian case is to  distinguish between the noncommutative multiplication on the
left and on the right.
%denoted by $\cal L$ and $\cal R$.
This requires building a symbolic representation, taking into account of the noncommutative multiplication,  for nonabelian difference polynomial ring, a ring of difference operators, quasi-local extensions of rings and formal pseudo-difference series with quasi-local coefficients.

Similar to the abelian scenario in \cite{MNWDiff},  the integrability conditions for \eqref{begeq} are given in terms of the coefficients of a canonical formal recursion operator. We prove that if an equation possesses an infinite dimensional algebra of symmetries, irrespective of their orders, then it possesses a unique canonical formal recursion operator. The coefficients of the canonical formal recursion operator can be explicitly determined  in terms of the symbolic representation of the equation. Moreover, if the equation is integrable, these coefficients must be quasi-local. We apply the developed formalism to classify integrable quasi-linear skew-symmetric equations (\ref{begeq}) of orders $(-1,1), (-2,2)$ and $(-3,3)$.

We start the paper with the fundamental algebraic framework necessary for investigating nonabelian evolutionary D$\Delta$Es. In Section \ref{prem} we introduce essential definitions and concepts of nonabelian difference algebra. This includes the grading of the algebra, derivations, evolutionary derivations, as well as algebras of local difference operators and formal series.
We then define algebras of symmetries and approximate symmetries as well as the notion of integrability and approximate integrability in Section \ref{symap}.

The symbolic representation serves as the primary computational tool. In Section \ref{symbrep} we define the symbolic representation of the nonabelian difference algebra, difference operators and formal series. We formulate criteria of approximate integrability in the symbolic representation. For a given evolutionary equation we provide a recursive formula for the coefficients of its symmetries (Theorem \ref{thesym}). We show that these coefficients are uniquely determined by the linear part of the symmetry. This result can be used to determine the existence of
fixed-order symmetries as well as to derive the integrability conditions if the orders of symmetries are known.

The main result on universal integrability conditions for evolutionary D$\Delta$Es on free associative algebras is Theorem \ref{themain} presented in Section \ref{intconds}.
The universal integrability conditions are the conditions on the coefficients of the canonical formal recursion operator: the coefficients must be quasi-local.
 The proof is based on computation of fractional powers of formal series with local coefficients. It is shown that coefficients of fractional powers of local formal series are quasi-local, i.e. belong to the appropriate algebraic extension of the algebra of formal series.
Furthermore, for a given equation these coefficients are recursively determined in terms of the symbolic representation of the equation (Theorem \ref{fropthe}).

In Section \ref{classres} we apply the developed theory to classify integrable quasi-linear skew-symmetric equations (\ref{begeq}) of orders $(-1,1), (-2,2)$ and $(-3,3)$. We list only equations which do not possess symmetries of lower orders. We prove integrability of each of these equations either by providing their Lax representations or by giving an explicit Miura type transformation to a known integrable equation. We present below two examples of new equations, to the best of our knowledge, resulting from the classification:
\begin{equation}
\label{BGnewint}
u_t=uu_1\cdots u_n(u+\alpha)-(u+\alpha)u_{-n}\cdots u_1u,\quad n\in\N
\end{equation}
where $\alpha$ is constant, and
\begin{equation}
\label{eq30m}
u_t=(u  u_{-1} + 1)  (u  u_{1} + 1)  u_{2} -
 u_{-2}  (u_{-1}  u + 1)  (u_{1}  u + 1).
\end{equation}
Equation (\ref{BGnewint}) possesses the Lax representation $L_t=[A,L]$,
$$
L=(u+\alpha)S^{-n}+\lambda uS,\quad A=uu_1\cdots u_n (1+\lambda S^{n+1}),
$$
where $\alpha$ is an essential constant and cannot be removed via invertible transformations.
The reduction $\alpha=0$ of this equation appears in \cite{BG}. 

Equation (\ref{eq30m}) possesses the Lax representation
$$
L=P^{-1}Q,\quad P=(1-S^{-2})u_1^{-1},\quad Q=\left((uu_1+1)(u_2u_1+1)S^2-1\right)u_{-1}^{-1}
$$
$$
A=(u_1u+1)(u_1u_2+1)S^2-u_{-1}\left((uu_1+1)(u_2u_1+1)-1\right)u_1^{-1}-S^{-2}.
$$
The abelian version of this equation was found and studied in \cite{GMY14}.

We conclude the paper by provide a brief summary of our results and concluding remarks on the implications and potential directions for future studies.

\section{Derivations, difference operators and formal series}\label{prem}
In this section, we introduce basic definitions and notations of objects and constructions for nonabelian difference algebras required for the aims of this paper. The detailed abelian analogue can be found, for example, in \cite{MNWDiff}. Further details regarding Hamiltonian structures of nonabelian D$\Delta$Es, as well as their connection to the problem of quantization of integrable systems, can be found in \cite{CW, CMW}.
\subsection{Difference algebra}\label{sec21}
Let $\A=\langle u_n,\,n\in\Z\rangle$ be a free associative unital algebra over $\C$ generated by an infinite set of non-commuting variables $u_n,\,n\in\Z$, and a unit element $e$, defined as
$$
e\cdot f(u_p,\ldots,u_q)=f(u_p,\ldots,u_q)\cdot e=f(u_p,\ldots,u_q),\quad f(u_p,\ldots,u_q)\in\A.
$$
Elements of algebra are polynomials (over $\C$) in variables $u_n,\,n\in\Z$. We use the notation $\alpha\cdot e\equiv\alpha,\,\,\alpha\in\C,$
and often write $u$ instead of $u_0$ without causing a confusion.

We define on $\A$ an automorphism $S:\,\A\to\A$ by the action
$$
S(f(u_p,\ldots,u_q))=f(u_{p+1},\ldots,u_{q+1}),\quad S(e)=e,\quad f(u_p,\ldots,u_q)\in\A,
$$
and call the automorphism $S$ the {\it shift operator}. The algebra $\A$ together with the shift operator $S$  is a difference algebra.

Let $T$ and $I$ denote the transposition and involution automorphisms respectively, defined as follows:
\begin{equation}\label{trans}
T(u_n)=u_n,\quad T(ab)=T(b)T(a),\quad T(e)=e,\quad a,b,\in\A
\end{equation}
and
\begin{equation*}
I(u_n)=u_{-n},\quad I(ab)=I(a)I(b),\quad I(e)=e,\quad a,b,\in\A.
\end{equation*}

Thus their composition $\cT=TI$ yields:
\begin{equation}
\label{taut}
\cT(u_n)=u_{-n},\quad \cT(ab)=\cT(b)\cT(a),\quad \cT(e)=e,\quad a,b,\in\A.
\end{equation}
The automorphism $\cT$ induces a $\Z_2$ grading on $\A$:
$$
\A=\A^0\oplus \A^1,\quad \A^0\cdot \A^0=\A^0,\quad \A^0\cdot \A^1=\A^1\cdot \A^0=\A^1,\quad \A^1\cdot \A^1=\A^0,
$$
where
$$
\A^i=\{a\in\A\,|\, \cT(a)=(-1)^i a\},\quad i=0,1.
$$
We refer to elements of $\A^0$ as {\it symmetric} and elements of $\A^1$ as {\it skew-symmetric}.

%In addition to a free associative difference algebra $\A$ we define its abelian analogue $\cA$. Let $\cA$ be an abelian unital algebra over $\C$ generated by variables $\cu_n$ and the unit element $\bf{1}$. We define the shift operator $S$ by the same  action 
%$$
%S(f(\cu_p,\ldots,\cu_q))=f(\cu_{p+1},\ldots,\cu_{q+1}),\quad S({\bf{1}})={\bf{1}},\quad f(\cu_p,\ldots,\cu_q)\in\cA.
%$$

%Denote by $\iota$ the abelianisation homomorphism $\iota: \A\to\cA$ defined by
%$$
%\iota(u_n)=\cu_n,\quad \iota(e)={\bf 1},
%$$
%so $\iota(f(u_p,\ldots,u_q))$ is a polynomial over $\C$ in commutative variables $\cu_p,\ldots,\cu_q$. The algebra $\cA$ together with the shift operator $S$ is an abelian difference algebra over $\C$.

\subsection{Derivations}
A derivation $\D$ of the algebra $\A$ is a $\C$-linear map satisfying the Leibnitz's rule
$$
\D(\alpha f+\beta g)=\alpha\D(f)+\beta\D(g),\quad \D(fg)=\D(f)g+f\D(g),\quad f,g\in\A,\quad\alpha,\beta\in\C.
$$
For any derivation $\D$ of $\A$ we have $\D(e)=0$ for the derivation of the unit in $\A$.
Each derivation can be uniquely defined by its action on the generators $u_n$ of $\A$.

For any two derivations $\D,\D'$ their commutator $[\D,\D']=\D\circ\D'-\D'\circ\D$  is also a derivation on $\A$. Any $\C$-linear combination of derivations is a derivation, and any triple of derivations $\D,\D',\D''$ satisfies the Jacobi identity
$$
[\D,[\D',\D'']]+[\D'',[\D,\D']]+[\D',[\D'',\D]]=0,
$$
and thus the set of derivations of $\A$ is a Lie algebra.

Define derivations $X_k,\,k\in\Z$ as
$$
X_k(u_j)=\delta_{kj}u_k,
$$
and let
$$
X=\sum_{k\in\Z}X_{k}
$$
The difference algebra $\A$ has a natural grading
\begin{equation}
\label{agr}
\A=\bigoplus_{p\ge 0}\A_p,\quad \A_p=\{f\in\A\,|\, X(f)=pf\}.
\end{equation}
Every element $f\in \A$ can be uniquely represented as a sum of homogeneous components $f=\sum_{p\ge 0}f_p,\,\,f_p\in\A_p$ (some components may be equal to zero).
The subspace $\A_0$ consists of elements of the form $\alpha\cdot e \equiv\alpha,\,\,\alpha\in\C$, while the subspaces $\A_1,\A_2,\ldots$ consist of linear, quadratic etc polynomials in generators $u_k$.
The projection $\pi_k:\A\to\A_k$ selecting $k$-th homogeneous component is defined as
\begin{equation}
\label{pik}
\pi_k(f)=\pi_k(\sum_{p\ge 0}f_p)=f_k.
\end{equation}

An important class of derivations of a difference algebra $\A$ are {\it evolutionary} derivations:
\begin{Def}
A derivation $\D$ is called {\it evolutionary} if it commutes with the shift operator $S$. 
\end{Def}
An evolutionary derivation is completely defined by its action on the generator $u$, that is
$$
\D(u)=a,\quad \D(u_k)=S^k(a),\quad a\in\A.
$$
The element $a\in\A$ is called the characteristic of an evolutionary derivation. We denote by $\D_a$ an evolutionary derivation with the characteristic $a$.

A commutator of two evolutionary derivations $\D_a$ and $\D_b$ is also an evolutionary derivation $\D_c=[\D_a,\D_b]$ with the characteristic $c=\D_a(b)-\D_b(a)$. Evolutionary derivations form a Lie subalgebra of the Lie algebra of derivations on $\A$.

The notion of an evolutionary derivation allows to introduce a Lie algebra structure on $\A$: for every $a, b\in\A$, we define a Lie bracket as
$$
[a,b]=\D_a(b)-\D_b(a).
%=\sum_i \left(m(S^i(a)\star_1\frac{\partial b}{\partial u_i})-m(S^i(b)\star_1\frac{\partial a}{\partial u_i})\right).
$$
It can be rewritten in terms of the Fr\'echet derivatives.
\begin{Def} A Fr\'echet derivative of an element $f\in\A$ is defined as
\begin{equation}\label{frech}
f_*(a)=\D_a(f),\quad\forall a\in\A.
\end{equation}
\end{Def}
Thus, the Lie bracket becomes
\begin{equation}\label{liebracket}
[a,b]=b_*(a)-a_*(b),\quad a,b\in\A.
\end{equation}
It is easy to see that the Lie algebra $\A$ is graded with respect to the natural grading (\ref{agr}):
$$
[\A_n,\A_m]\subset\A_{n+m-1}.
$$

\subsection{Difference operators and formal difference series}\label{Sec23}

%We first recall the notion of local operators and formal series in the abelian case.  Let $\cA$ be an abelian difference algebra  (one may simply assume that in the above construction we set the multiplication to be commutative).  In the abelian case one considers the notion of local operators of order $\ord\, A:=(p,q)$ as objects of the form
%$$
%A=\sum_{i=p}^qa_iS^i,\quad a_i\in\cA,\quad p,q\in\Z.
%$$
%The total order of an operator is defined as $\Ord\, A:=p-q$.
%More generally one also considers local formal difference series (or just formal series) as objects 
%$$
%A=\sum_{i\le p}a_{i}S^i,\quad a_{i}\in\cA,\quad p\in\Z.
%$$
%Local difference operators/formal series on $\cA$ form a unital algebra $\cA((S))$ with addition defined as
%\begin{equation}
%\label{opadd}
%A=\sum_{i}a_{i}S^i,\quad B=\sum_{i}b_{i}S^i,\quad A+B=\sum_i(a_{i}+b_{i})S^i,\quad a_{i},b_{i}\in\cA
%\end{equation}
%and the multiplication defined on monomials as
%$$
%aS^n\cdot b S^m=aS^n(b)S^{n+m}.
%$$
%This multiplication is associative but not commutative. 

%The algebra $\cA((S))$ inherits the natural grading 
%$$
%\cA=\bigoplus_{k\ge 0}\cA_k((S)),\quad \cA_p((S))\cdot\cA_q((S))\subset \cA_{p+q}((S))
%$$
%where each subspace $\cA_k$ consists of elements of the form $a_kS^i,\, a_k\in\cA_k,\,i\in\Z$. 

%If $A=\sum_{i=p}^qa_iS^i$ is a local difference operator on $\cA$ then the action of $A$ on an element $f\in\cA$ is defined as
%$$
%A(f)=\sum_{i=p}^qa_iS^i(f)\in\cA.
%$$
%Note that the action of a formal series on $\cA$ is generally not well-defined.

In this section, we extend the notion of difference operators/formal series (for the abelian case, refer to \cite{MNWDiff}) to the non-commutative case and introduce the algebra of formal series $\A((S))$.

Let $\cL_a$ and $\cR_a$ be the left and right multiplication operators on the algebra $\A$, defined as:
$$
\cL_a(f)=af\quad \mbox{and} \quad \cR_a(f)=fa,
$$
where $a, f\in\A.$ These operators satisfy the following relations:
$$
\cL_a\cR_b=\cR_b \cL_a,\quad \cL_a\cL_b=\cL_{ab}\quad \mbox{and} \quad\cR_a\cR_b=\cR_{ba}
$$
for all $a, b \in \A$. We denote by $\cM$ the algebra of the left and right multiplication operators.
This algebra inherits the natural grading (\ref{agr}) on $\A$:
\begin{equation}
\label{mgr}
\cM=\bigoplus_{p,q\ge 0}\cM_{p,q},\quad \cM_{p,q}\cM_{p',q'}\subset\cM_{p+p',q+q'},
\end{equation}
where each element in the $\C$-linear subspace $\cM_{p,q}$ is a finite linear combination of terms of the form
$\cL_{a}\cR_{b}, \, a\in\A_p, b\in\A_q$. We represent it as
\begin{eqnarray}\label{mpq}
\cM_{p,q}=\Big\{\sum_{\gamma=1}^k\cL_{a^{(\gamma)}}\cR_{b^{(\gamma)}}\,|\,a^{(\gamma)}\in\A_p, b^{(\gamma)}\in\A_q, k\in \N\Big\}.
\end{eqnarray}
\begin{Def}
A {\it local difference operator} of order $\ord\, B:=(p,q)$ is of the form
$$
B=\sum_{i=p}^qb^{(i)}S^i,\quad b^{(i)}\in\cM,\quad b^{(p)}\neq 0, \ b^{(q)}\neq 0, \quad p,q\in\Z, \quad p\leq q
$$
and its total order is  $\Ord\, B:=p-q$.
A {\it local formal difference series} $B$ is defined as
\begin{equation}\label{series}
B=\sum_{i\le p} b^{(i)}S^i, \quad b^{(i)}= \sum_{j,k\ge 0} b_{jk}^{(i)} \in \cM,
\quad b_{jk}^{(i)}\in\cM_{j,k} \quad b^{(p)}\neq 0,  \quad p\in\Z.
\end{equation}
\end{Def}
%We can extend the action of the abelianisation homomorphism on the set of local difference formal series by setting
%$$
%\iota(\cL_f\cR_gS^n)=\iota(f)\iota(g)S^n,
%$$
%and thus $\iota$ maps the set of local formal series over a free associative difference algebra $\A$ onto the set of local formal series over its abelian analogue.

The action of a local difference operator on $\A$ is well-defined if we set
$$
\cL_f\cR_gS^n(h)=fS^n(h)g,\quad f,g,h\in\A, \ n\in \Z.
$$
The action of a formal series on $\A$ is generally not well-defined. 

The Fr\'echet derivative of any element in $\A$, as defined by \eqref{frech}, serves as an example of a local difference operator. For a non-constant element $f\in\A$, we define its order and total order as those for $f_*$.
\begin{Ex} Consider the nonabelian Volterra equation
\begin{equation}\label{volt}
u_t=uu_1-u_{-1}u=f.
\end{equation}
Note that
$$
f_*=\cL_u S+\cR_{u_1}-\cL_{u_{-1}}-\cR_u S^{-1} .
$$
Thus the nonabelian Volterra equation has the order $(-1,1)$ and the total order $2$.
\end{Ex}
Local formal series on $\A$ form a unital associative non-commutative algebra $\A((S))$ with the standard addition
\iffalse
defined as
\begin{equation}
\label{opadd}
A=\sum_{i\le p}\sum_{j,k\ge 0}a_{ijk}S^i,\quad B=\sum_{i\le p}\sum_{j,k\ge 0}b_{ijk}S^i,\quad A+B=\sum_{i\le p}\sum_{j,k\ge 0}(a_{ijk}+b_{ijk})S^i,\quad a_{ijk},b_{ijk}\in\cM_{j,k},
\end{equation}
\fi
and the multiplication rule (defined on monomials) given by
$$
\cL_f\cR_gS^n\cdot\cL_{f'}\cR_{g'}S^m=\cL_{fS^n(f')}\cR_{S^n(g')g}S^{n+m},\quad f, g, f', g'\in\A,\quad \ n, m\in \Z.
$$
The action of an evolutionary derivation $\D_a$ on an element $\cL_f\cR_gS^i$ in $\A((S))$ is defined as
\begin{equation}\label{derser}
\D_a(\cL_f\cR_gS^i)=\cL_{\D_a(f)}\cR_gS^i+\cL_f\cR_{\D_a(g)}S^i=\cL_{f_*(a)}\cR_gS^i+\cL_f\cR_{g_*(a)}S^i .
\end{equation}
Thus it is well-defined on any formal series in $\A((S))$.

The algebra $\A((S))$ inherits the natural grading \eqref{mgr} on $\cM$:
\begin{equation}
\label{asgr}
\A((S))=\bigoplus_{p,q\ge 0}\A_{p,q}((S)),\quad \A_{p,q}((S))\cdot\A_{p',q'}((S))\subset\A_{p+p',q+q'}((S)),
\end{equation}
where each $\C$-linear subspace $\A_{p,q}((S))$ consists of formal series defined by (\ref{series}) when $b^{(i)}\in\cM_{p,q}$.

\section{Symmetries and approximate symmetries}\label{symap}
We begin with defining an evolutionary D$\Delta$Es on free associative algebras.
Assume that the generators $u_k$ of the free associative algebra $\A$ depend on $t\in\C$. We identify every evolutionary derivation $\D_f$ of $\A$ with a D$\Delta$E on $\A$, namely,
\begin{equation}
\label{dd}
u_{t}=\D_f(u)=f, \quad f\in\A.
\end{equation}
In particular, we have $u_{n,t}=S^n(f),\,\,n\in\Z$. The order of a D$\Delta$E is the order of $f$, that is, the order of its Fr\'echet derivative $f_*$.
Evolution with respect to $t$ of any element $a\in\A$ is therefore given by $a_t=\D_f(a)=a_*(f)$.

We call an evolutionary equation \eqref{dd} {\it skew-symmetric} if $f$ is skew-symmetric, i.e., $f\in \A^1$
satisfying $\cT(f)=-f$ under the action of automorphism $\cT$ defined by (\ref{taut}).
\subsection{Symmetries of evolutionary D$\Delta$Es}
%In what follow we study generators of infinitesimal symmetries of evolutionary D$\Delta$Es and for brevity call them symmetries.
\begin{Def} \label{defsym1}
We say that $g\in\A$ is a symmetry of a D$\Delta$E (\ref{dd}) if the Lie bracket \eqref{liebracket} between $f$ and $g$ vanishes, that is, $[g,f]=0$.
%The order of the symmetry $g$ is the order of the difference operator $g_*$.
\end{Def}
For every element $g\in\A$, we associate an evolutionary derivation $\D_g$. The condition that $g$ is a symmetry of (\ref{dd}) is equivalent to the evolutionary derivation $\D_g$ commuting with the evolutionary derivation $\D_f$, i.e. $[\D_f,\D_g]=0$. If we assume that the generators $u_k$ also depend on $\tau\in\C$, the conventional way of representing a symmetry of equation (\ref{dd}) is to associate the evolutionary derivation $\D_g$ with a D$\Delta$E:
$$
u_{\tau}=\D_g(u)=g.
$$
A linear combination $\alpha g_1+\beta g_2$, where $\alpha,\beta\in\C$, of two symmetries $g_1$ and $g_2$, as well as their Lie bracket $[g_1,g_2]$ are also symmetries of the equation (\ref{dd}). Thus the set of all symmetries of the equation (\ref{dd}), denoted by
$$
\cC_{f}=\{g\in\A\,|\,[g,f]=0\},
$$
forms a Lie subalgebra of $\A$. This subalgebra is the centraliser of $f\in\A$.

The characteristic property of integrable systems is the existence of an infinite-dimensional Lie algebra of symmetries. Therefore, we adopt the following definition of integrability:
\begin{Def}\label{defsym} A nonabelian D$\Delta$E (\ref{dd}) is called integrable if its Lie algebra of symmetries $\cC_{f}$ is infinite dimensional and contains symmetries of arbitrarily high total order, in the sense that for every $N\in\N$, the equation possesses a symmetry $g$ of total order $\Ord\,\, g\geq N$.
\end{Def}

\begin{Ex}\label{EG2} For the nonabelian Volterra equation \eqref{volt},
one can directly verify that
$$
u_{\tau}=uu_1u_2-u_{-2}u_{-1}u+u u_1^2 +u^2u_1-u_{-1}^2 u-u_{-1}u^2
$$
is a symmetry of order $(-2,2)$ and the total order $4$. It is known that the nonabelian Volterra equation possesses infinitely many symmetries of order $(-n,n),\,n\in\N$. These can be explicitly found either from the Lax representation \cite{BG} or by means of the recursion operator \cite{CasatiWang}.
\end{Ex}

\subsection{Approximate symmetries}
Based on the gradation on $\A$, we shall introduce the notion of approximate symmetries.

Consider a nonabelian D$\Delta$E \eqref{dd}
\iffalse
\begin{equation}
\label{eqgen}
u_t=f,\quad f\in\A,
\end{equation}
\fi
and represent $f$ as a sum of homogeneous components
$$
f=f_0+\cdots+f_N,\quad f_{k}\in\A_k,\quad k=0,\ldots,N.
$$
According to Definition \ref{defsym1} the element $g=g_0+\cdots+g_M,\,\,g_i\in\A_i,\,i=0,\ldots,M$ is a symmetry if $[f,g]=0$. Due to the natural gradation on the Lie algebra $\A$, the vanishing of the Lie bracket is equivalent to
\begin{equation}
\label{appsym}
\sum_{k=0}^p[f_k,g_{p-k}]=0,\quad p=0,1,\ldots,N+M,
\end{equation}
where we assume $f_k=0,\,k>N$ and $g_{k'}=0,\,k'>M$. We call $g$ {\it an approximate symmetry of degree $s$} or {\it $s$-approximate symmetry} if the first $s+2$ relations of \eqref{appsym} with $0\le p\le s+1$ are satisfied. In other words, $g$ is an approximate symmetry of degree $s$ if
$$
\pi_k([f,g])=0,\quad k=0,\ldots, s.
$$

Similar to symmetries, all $s$-approximate symmetries form a Lie subalgebra of $\A$:
$$
\cC_{f}^s=\{g\in\A\,|\,\pi_k([g,f])=0,\,\,k=0,1,\ldots,s\}.
$$

Note that there is no obstruction to replace the free associative algebra $\A$ by an algebra of formal series defined as
$$
\overline{\A}=\{\sum_{k=0}^{\infty}f_k\,|\,f_k\in\A_k\}.
$$
One can then consider equations and their symmetries defined by formal series in $\overline{\A}$. The definition of integrability, identical to the abelian case \cite{MNWDiff}, for formal series is as follows:

\begin{Def} A nonabelian D$\Delta$E (\ref{dd}) for $f\in\overline{\A}$ is called $s$-approximate integrable if its Lie algebra $\cC_{f}^s$ is infinite dimensional and contains $s$-approximate symmetries of arbitrarily high total order. The equation is called formally integrable if it is $s$-approximate integrable for arbitrarily large $s$.
\end{Def}
Any integrable equation is formally integrable, while the converse is not necessarily true as there is no assumption regarding the convergence of formal series and their approximate symmetries. Conditions of $s$-approximate integrability can be explicitly written and verified for a given equation. Moreover, these are necessary integrability conditions suitable for classification of integrable equations. The natural approach for formulating these conditions is through symbolic representation, a concept that will be elaborated upon in the subsequent sections.

In what follows we will consider nonabelian evolutionary D$\Delta$Es without a constant term.
A constant term in the equation can often be removed by an invertible change of variables. Moreover, if an equation possesses $N$ symmetries with constant terms, then taking their linear combination one can always construct $N-1$ symmetries without a constant term.
%and therefore, it is sufficient to consider a subalgebra of symmetries without a constant term.
We denote the difference algebra without constant terms by $\A'=\bigoplus_{k\ge 1}\A_k$, and similarly the difference algebra of formal series without constant terms by $\overline{\A}'=\{\sum_{k=1}^{\infty}f_k\,|\,f_k\in\A_k\}$. Hereafter, we restrict ourselves to nonabelian D$\Delta$Es of the form
\begin{equation}\label{ddn}
 u_t=f, \quad f\in\A' \ \mbox{or} \ f\in\overline{\A}'.
\end{equation}

\section{Symbolic representation}\label{symbrep}

Symbolic representation was first applied in the area of integrable systems by Gel'fand and Dickey \cite{mr58:22746} and later further developed in the framework of the symmetry approach in works of Beukers, Sanders, Wang, Mikhailov, Novikov and van der Kamp \cite{mr99i:35005}-\cite{vdk2}. In \cite{MNWDiff}
the method of symbolic representation was extended to the abelian difference algebra as well as difference operators and formal series. Now we extend the approach to the nonabelian case.

\subsection{Symbolic representation of difference algebra $\A$}

To define the symbolic representation $\hat{\A}=\bigoplus_{p\ge 0}\hat{\A}_p$ of the naturally graded difference algebra $\A=\bigoplus_{p\ge 0}\A_p$, we first define an isomorphism of $\C$-linear spaces  $\phi: \A_p\to\hat{\A}_p$, and then extend it to difference algebra and Lie algebra  isomorphism by defining the multiplication, the action of the shift operator, derivations and Lie bracket on $\hat{\A}$. The isomorphism of linear spaces is uniquely determined by its action on monomials.
\begin{Def}
\label{symbdef}
The symbolic representation of difference monomial terms is defined as
$$
\phi:\,e\mapsto1,\quad\quad \phi:\,u_{i_1}u_{i_2}\cdots u_{i_n}\in\A_n\mapsto\hu^n\xi_1^{i_1}\xi_2^{i_2}\cdots\xi_n^{i_n}\in\hat{\A}_n .
$$
\end{Def}
Notations $\xi_1,\xi_2,\ldots$ are reserved for the variables in the symbolic representation, while $\hu^n$ indicates the appropriate linear space $\hat{\A}_n$ to which the symbol belongs to. We note here that unlike the situation in the abelian case \cite{MNWDiff} the nonabelian symbolic representation does not require symmetrisation operation.

\begin{Ex}\label{eg2} Here we give a few concrete examples:
\begin{eqnarray*}
&&u_k\raph\hu\xi_1^k,\quad u^m\raph\hu^m,\quad u_1u_{-1}\raph\hu^2\xi_1\xi_2^{-1},\\
&&\alpha uu_1u_2+\beta u_{-2}u_{-1}u\raph \hu^3(\alpha\xi_2\xi_3^2+\beta\xi_1^{-2}\xi_2^{-1}),\quad \alpha,\beta\in\C.
\end{eqnarray*}
\end{Ex}
Under the action of the isomorphism $\phi$ a homogeneous polynomial $f\in\A_n$ is in one-to-one correspondence with a symbol $\hat{f}=\hu^n a(\xi_1,\ldots,\xi_n)$, where the coefficient function $a(\xi_1,\ldots,\xi_n)$ is a Laurent polynomial in variables $\xi_1,\ldots,\xi_n$.

The action of automorphisms $S, T, I$ and $\cT$, presented in Section \ref{sec21}, on $f\in\A_n$ is defined by
\begin{eqnarray*}
&&S(f)\raph \hu^n a(\xi_1,\ldots,\xi_n)\xi_1\cdots\xi_n;\quad T(f)\raph \hu^n a(\xi_n,\ldots,\xi_1);\\
&&I(f)\raph \hu^n a(\xi_1^{-1},\ldots,\xi_n^{-1});\qquad\qquad \cT(f)\raph \hu^n a(\xi_n^{-1},\ldots,\xi_1^{-1}).
\end{eqnarray*}
The projector $\pi_k$ selects the $k$-th homogeneous component of $f\in\A$. Its symbolic representation $\hat{\pi}_k$ is induced by the condition $\phi\pi_k=\hat{\pi}_k\phi$. If
$$
f=\sum_{l\ge 0}f_l\in\overline{\A},\quad\pi_k(f)=f_k\in\A_k,\quad\phi(f_k)=\hu^k a_k(\xi_1,\ldots,\xi_k),
$$
then $\hat{f}:=\phi(f)=\sum_{l\ge 0}\hu^l a_l(\xi_1,\ldots,\xi_l)$ and
$$
\hat{\pi}_k(\hat{f})=\hat{\pi}_k\left(\sum_{l\ge 0}\hu^l a_l(\xi_1,\ldots,\xi_l)\right)=\hu^k a_k(\xi_1,\ldots,\xi_k).
$$

We now define the multiplication in the symbolic representation.
\begin{Def} Let $f\in\A_n,\,g\in\A_m$ and $\phi(f)=\hu^na(\xi_1,\ldots,\xi_n)\,\,\phi(g)=\hu^mb(\xi_1,\ldots,\xi_m)$. Then $\phi(fg)=\phi(f)*\phi(g)$, where
\begin{equation}
\label{multrule}
\phi(f)*\phi(g)=\hu^{n+m}a(\xi_1,\ldots,\xi_n)b(\xi_{n+1},\ldots,\xi_{n+m}).
\end{equation}
\end{Def}
Note that the symbolic representation of difference monomials (Definition \ref{symbdef}) can be deduced from the representation of linear monomials $\phi(u_k)=\hu\xi_1^k$ and the multiplication rule (\ref{multrule}).
The commutative algebra $\hat{\A}=\oplus\hat{\A}_k$ equipped with the $*$ multiplication is isomorphic to the difference algebra $\A$.

\subsection{Symbolic representation of difference operators and formal series}\label{diffopsymb}
We extend the symbolic representation to difference operators and formal series, thereby obtaining the symbolic representation $\hat{\A}((\eta))$ for the algebra of formal series $\A((S))$.

We assign the symbol $\eta$ to the symbolic representation of the shift operator $S$ with the action rule given by
$$
\eta(\hu^na(\xi_1,\ldots,\xi_n))=\hu^na(\xi_1,\ldots,\xi_n)\xi_1\cdots\xi_n
$$
and the composition rule, corresponding to $S\circ f=S(f)S$, defined as
$$
\eta\circ\hu^na(\xi_1,\ldots,\xi_n)=\hu^na(\xi_1,\ldots,\xi_n)\xi_1\cdots\xi_n\eta.
$$

Consider $A=\cL_f\cR_gS^i\in\A_{n,m}((S))$, where the grading of $\A((S))$ is defined by \eqref{asgr} and
$$\phi(f)=\hat{f}=\hu^n a(\xi_1,\ldots,\xi_n),\quad  \phi(g)=\hat{g}=\hu^m b(\xi_1,\ldots,\xi_m).$$ Then its symbolic representation $\hat{A}$ is
$$
 A=\cL_f\cR_gS^i\,\,\raph\,\, \hat{A}=\hu^n_l\hu^m_r a(\xi_1,\ldots,\xi_n) b(\zeta_1,\ldots,\zeta_m)\eta^i.
$$
Here we reserve the notation $\hu^n_l a(\xi_1,\ldots,\xi_n)$ for the symbol of the left multiplication operator $\cL_f$, while $\hu^m_r b(\zeta_1,\ldots,\zeta_m)$ denotes the symbol of the right multiplication operator $\cR_g$. Referring to \eqref{mpq}, when dealing with an operator $b^{(i)} S^i$, where $b^{(i)}\in \cM_{p,q}$, we adopt a shortened notation for its symbolic representation:
$$
\sum_{\gamma=1}^k \hu^n_l\hu^m_r a^{(\gamma)}(\xi_1,\ldots,\xi_n) b^{(\gamma)}(\zeta_1,\ldots,\zeta_m)\eta^i=:\hu^n_l\hu^m_r a(\xi_1,\ldots,\xi_n,\zeta_1,\ldots,\zeta_m)\eta^i.
$$
The symbolic representation $\hat{\A}_{p,q}((\eta))$ of a subspace $\A_{p,q}((S))$ is the sum of the terms in the form
$
\hu^p_l\hu_r^q a^{(s)}(\xi_1,\ldots,\xi_p,\zeta_1,\ldots,\zeta_q)\eta^s,
$
where $s\leq N_{pq}$ and $ N_{pq}\in \Z$. We simply denote such element in $\hat{\A}_{p,q}((\eta))$ by
$\hu_l^p\hu_r^{q}a(\xi_1,\ldots,\xi_p,\eta,\zeta_1,\ldots,\zeta_q)$.
Thus, for a generic formal series $B\in \A((S))$, following from \eqref{asgr}, its symbolic representation $\hat{B}$ can be written as
\begin{equation}
\label{fs}
\hat{B}=\sum_{p\ge 0}\sum_{q\ge 0}\hu_l^p\hu_r^{q}a_{pq}(\xi_1,\ldots,\xi_p,\eta,\zeta_1,\ldots,\zeta_q),
\end{equation}
where the summation indices $p,q$ represent the summation over elements belonging to different subspaces $\hat{\A}_{p,q}((\eta))$, and
\begin{eqnarray*}
&&\hu_l^p\hu_r^{q}a_{pq}(\xi_1,\ldots,\xi_p,\eta,\zeta_1,\ldots,\zeta_q)=\hu_l^p\hu_r^{q}\!\sum_{s\le N_{pq}}a_{pq}^{(s)}(\xi_1,\ldots,\xi_p,\zeta_1,\ldots,\zeta_q)\eta^s\!\in \hat{\A}_{p,q}((\eta)), N_{pq}\in\Z,
\end{eqnarray*}
where $a_{pq}^{(s)}$ are Laurent polynomials in  $\C[\xi_1^{\pm 1},\ldots,\xi_p^{\pm 1},\zeta_1^{\pm 1},\ldots,\zeta_q^{\pm 1}]$.

Naturally, the projector $\hat{\pi}_{p,q}:\hat{\A}((\eta))\to\hat{\A}_{p,q}((\eta))$  on formal series in the symbolic representation becomes
$$
\hat{\pi}_{p,q}(\hat{B})=\hu_l^p\hu_r^qa_{pq}(\xi_1,\ldots,\xi_p,\eta,\zeta_1,\ldots,\zeta_q).
$$

We now define the operations in the symbolic representation $\hat{\A}((\eta))$.
To the sum of formal series in $\A((S))$ corresponds the sum of their symbolic representations in $\hat{\A}((\eta))$.
The composition rule $\cL_f\cR_gS^i\circ\cL_{f'}\cR_{g'}S^{i'}=\cL_{fS^{i}(f')}\cR_{S^{i}(g')g}S^{i+i'}$ in the symbolic representation takes the form
\begin{eqnarray}
\label{comp}
&&\hu^n_l\hu^m_ra(\xi_1,\ldots,\xi_n,\eta,\zeta_1,\ldots,\zeta_m)\circ \hu^{n'}_l\hu^{m'}_ra'(\xi_1,\ldots,\xi_{n'},\eta,\zeta_1,\ldots,\zeta_{m'})\\ \nonumber
&&=\hu_l^{n+n'}\hu_r^{m+m'}a(\xi_1,\ldots,\xi_n,\eta\prod_{j=n+1}^{n+n'}\!\!\xi_j\prod_{k=1}^{m'}\!\zeta_k,\zeta_{m'+1},\ldots,\zeta_{m+m'})a'(\xi_{n+1},\ldots,\xi_{n+n'},\eta,\zeta_1,\ldots,\zeta_{m'}).
\end{eqnarray}
By linearity we extend this rule to the composition rule of formal series in $\hat{\A}((\eta))$.

Next we state the symbolic representations for the Fr\'echet derivative and the action of difference operators on ${\A}$.
Due to linearity of these operations, we only represent them for an element in subspace $\A_p, 0\leq p\in \mathbb{Z}$.

Let $f\in\A_p$ with a symbol  $\hat{f}=\hu^p b(\xi_1,\ldots,\xi_p)$. The symbolic representation of its Fr\'echet derivative is given by
\begin{equation}\label{symfre}
 \hat{f}_*=\sum_{i=1}^p\hu_l^{i-1}\hu_r^{p-i} b(\xi_1,\ldots,\xi_{i-1},\eta,\zeta_1,\ldots,\zeta_{p-i}),
\end{equation}
The action rule of a difference operator on $f\in\A_p$ is derived from the action
\begin{eqnarray}\label{opalg}
&&\quad \hu_l^n\hu_r^ma(\xi_1,\ldots,\xi_n,\eta,\zeta_1,\ldots,\zeta_m)\left(\hu^pb(\xi_1,\ldots,\xi_p)\right)\nonumber\\
&&=\hu^{n+m+p}a(\xi_1,\ldots,\xi_n,\prod_{i=n+1}^{n+p}\xi_i,\xi_{n+p+1},\ldots,\xi_{n+m+p}) b(\xi_{n+1},\ldots,\xi_{n+p}).
\end{eqnarray}
\iffalse
and, extended by linearity, the rule \eqref{opalg} provides the way of computation of the symbol of a Fr\'echet derivative of a generic element of the difference algebra $\A$.
\fi
\begin{Ex}
Let $f=uu_1u_2-u_{-2}u_{-1}u$. We know that $\hat{f}=\hu^3 (\xi_2 \xi_3^2-\xi_1^{-2} \xi_2^{-1}) $ from Example \ref{eg2}. It follows from \eqref{symfre} that
$$
\hat{f}_*=\phi(f_*)=\hu^2_r\zeta_1\zeta_2^2+\hu_l\hu_r\zeta_1^2\eta+\hu_l^2\xi_2\eta^2-\hu_r^2\zeta_1^{-1}\eta^{-2}-\hu_l\hu_r\xi_1^{-2}\eta^{-1}-\hu_l^2\xi_1^{-2}\xi_2^{-1},
$$
which is the symbol of
$$f_*=\cR_{u_1u_2}+\cL_u\cR_{u_2}S+\cL_{uu_1}S^2-\cR_{u_{-1}u}S^{-2}-\cL_{u_{-2}}\cR_{u}S^{-1}-\cL_{u_{-2}u_{-1}}.$$
\end{Ex}
By combining the symbolic representation of Fr\'echet derivative \eqref{symfre} with \eqref{opalg}, we obtain the symbolic representation of the Lie bracket \eqref{liebracket}.
\begin{Pro}
Suppose $f\in \A_n$ with the symbol $\hat{f}=\hu^n a(\xi_1,\ldots,\xi_n)$ and $g\in\A_m$ with the symbol $\hat{g}=\hu^m b(\xi_1,\ldots,\xi_m)$. Then
\begin{eqnarray}
\nonumber &&\phi([f,g])=\sum_{i=1}^m\hu^{m+n-1}b(\xi_1,\ldots,\xi_{i-1},\prod_{j=i}^{i+n-1}\xi_j,\xi_{i+n},\ldots,\xi_{m+n-1})a(\xi_i,\ldots,\xi_{i+n-1})\\ \label{brsymb}
&&\qquad -\sum_{i=1}^n\hu^{m+n-1}a(\xi_1,\ldots,\xi_{i-1},\prod_{j=i}^{i+m-1}\xi_j,\xi_{i+m},\ldots,\xi_{m+n-1})b(\xi_i,\ldots,\xi_{i+m-1}).
\end{eqnarray}
\end{Pro}
\begin{proof} Directly applying \eqref{symfre} and \eqref{opalg}, we obtain
\begin{eqnarray}
&&\phi(g_*(f))=\sum_{i=1}^m\hu_l^{i-1}\hu_r^{m-i} b(\xi_1,\ldots,\xi_{i-1},\eta,\zeta_1,\ldots,\zeta_{m-i})\big(\hu^n a(\xi_1,\ldots,\xi_n)\big)\nonumber\\
&&=\sum_{i=1}^m\hu^{m+n-1}b(\xi_1,\ldots,\xi_{i-1},\prod_{j=i}^{i+n-1}\xi_j,\xi_{i+n},\ldots,\xi_{m+n-1})a(\xi_{i},\ldots,\xi_{i+n-1}).\label{freact}
\end{eqnarray}
Then formula \eqref{brsymb} follows from the definition of Lie bracket \eqref{liebracket}.
\end{proof}
In particular, if $f\in\A_1$ with the symbol $\hat{f}=\hu\omega(\xi_1)$, then
\begin{equation}\label{brom}
\phi([f,g])=\hu^n b(\xi_1,\ldots,\xi_m)\Big(\omega(\xi_1)+\cdots+\omega(\xi_m)-\omega(\xi_1\cdots\xi_m)\Big).
\end{equation}
We conclude this section with the symbolic representation for
the action of an evolutionary derivation $\D_f$ on an element $B\in\A((S))$ defined by \eqref{derser}.

Let $\hat{f}=\hu^p a(\xi_1,\ldots,\xi_p)$ and $B\in\A((S))$ with $\phi(B)=\hu^n_l\hu_r^mb(\xi_1,\ldots,\xi_n,\eta,\zeta_1,\ldots,\zeta_1,\ldots,\zeta_m)$.
Using \eqref{freact}, we obtain
\begin{eqnarray}
&&\phi(\D_f(B))\!=\!\!\sum_{i=1}^n\hu_l^{n+p-1}\hu_r^mb(\xi_1,\ldots,\xi_{i-1},\!\prod_{j=i}^{i+p-1}\!\!\xi_j,\xi_{i+p},\ldots,\xi_{n+p-1},\eta,\zeta_1,\ldots,\zeta_m)a(\xi_i,\ldots,\xi_{i+p-1})\nonumber\\
&&\quad +\sum_{i=1}^n\hu_l^n\hu_r^{m+p-1}b(\xi_1,\ldots,\xi_n,\eta,\zeta_1,\ldots,\zeta_{i-1},\!\prod_{j=i}^{i+p-1}\!\!\zeta_j,\zeta_{i+p},\ldots,\zeta_{m+p-1})a(\zeta_i,\ldots,\zeta_{i+p-1}).\label{symderser}
\end{eqnarray}

\subsection{Symmetries in the symbolic representation}\label{symsymb}

We proceed to describe the symmetry algebra of a given equation in symbolic representation. We remind the reader that we consider nonabelian difference equations \eqref{ddn} without a constant term.  The findings presented in this section bear similarities to those in the abelian case \cite{MNWDiff}, albeit accounting for the distinction in symbolic representation between the abelian and nonabelian cases.

\begin{Pro}
\label{prolin}
Let $g\in\A'$ be a symmetry of equation \eqref{ddn} with a non-zero linear term,i.e. $f_1=\pi_1(f)\ne 0$. Then the symmetry $g$ also possesses a non-zero linear term $g_1=\pi_1(g)\ne 0$.
\end{Pro}
\begin{proof} Assume that the symmetry $g=g_k+g_{k+1}+\cdots,\,g_k\in\A_k$  and $g_k\neq 0$ for some $k>1$, which implies $\pi_k([f,\ g])=[f_1, \ g_k]=0$. Let the symbolic representation of $f_1$ be $\hu\omega(\xi_1)$ and the symbolic representation of  $g_k$ be $\hu^k b(\xi_1,\ldots,\xi_k)$. Then it follows from (\ref{brom}) that $ b(\xi_1,\ldots,\xi_k)=0$. Thus $g_k=0$ contradicting with the assumption.
\end{proof}

\begin{Pro} Let $u_t=f \in\A_1$ be a linear equation. Then the algebra of its symmetries coincides with $\A_1$, that is, $\cC_f=\A_1$.
\end{Pro}
\begin{proof} Let $g=\sum_{k\ge 1}g_k,\,g_k\in\A_k$ be a symmetry of the equation $u_t=f\in\A_1$. Then we must have $\pi_k([f,g])=0,\,k=2,3,\ldots$. Let $\hu^k b_k(\xi_1,\ldots,\xi_k)$ be the symbol of $g_k$. Then from (\ref{brom}) it follows that $b_k(\xi_1,\ldots,\xi_k)=0$, which implies $g_k=0$ for all $k\geq 2$.
\end{proof}

\begin{Pro} The Lie algebra of symmetries of equation $u_t=f\in\A'$ with a non-zero linear term $f_1=\pi_1(f)\ne 0$ is commutative.
\end{Pro}
\begin{proof} By Proposition \ref{prolin} a symmetry of  equation $u_t=f\in\A'$ with a non-zero linear term $f_1=\pi_1(f)\ne 0$ must contain a linear term. Let $g, h$ be two symmetries. Then $[g, h]$ is a symmetry without a linear term implying that $[g, h]=0$.
\end{proof}

For a given D$\Delta$E,  using symbolic representation, we are able to explicitly present its symmetries as follows:
\begin{The}\label{thesym} Consider a nonabelian D$\Delta$E
\begin{equation}
\label{eqsym}
u_t=f,\quad f\in\A',\quad \hat{f}=\hu\omega(\xi_1)+\sum_{k\ge 2}\hu^ka_k(\xi_1,\ldots,\xi_k),\quad\omega(\xi_1)\ne 0.
\end{equation}
Let $g\in\A'$ be a symmetry. Then the coefficients $A_k$ of its symbolic representation
\begin{equation}
\label{symsym}
\hat{g}=\hu\Omega(\xi_1)+\sum_{k\ge 2}\hu^kA_k(\xi_1,\ldots,\xi_k)
\end{equation}
can be determined recursively:
\begin{eqnarray}
\label{symcoef2}
&&A_2(\xi_1,\xi_2)=\frac{G^{\Omega}(\xi_1,\xi_2)}{G^{\omega}(\xi_1,\xi_2)}a_2(\xi_1,\xi_2),\\
\label{symcoefk}
&&A_k(\xi_1,\ldots,\xi_k)=\frac{1}{G^{\omega}(\xi_1,\ldots,\xi_k)} \left(
G^{\Omega}(\xi_1,\ldots,\xi_k) a_k(\xi_1,\ldots,\xi_k) \right.\\ \nonumber &&\qquad +\sum_{j=2}^{k-1}\sum_{i=1}^jA_j(\xi_1,\ldots,\xi_{i-1},\prod_{s=i}^{i+k-j}\xi_s,\xi_{i+k+1-j},\ldots,\xi_k)a_{k+1-j}(\xi_i,\ldots,\xi_{i+k-j})\\ \nonumber
&&\qquad \left.-\sum_{j=2}^{k-1}\sum_{i=1}^ja_j(\xi_1,\ldots,\xi_{i-1},\prod_{s=i}^{i+k-j}\xi_s,\xi_{i+k+1-j},\ldots,\xi_k)A_{k+1-j}(\xi_i,\ldots,\xi_{i+k-j}) \right),
\end{eqnarray}
where $k\geq 3$ and
\begin{equation}
\label{Gw}
G^\varkappa (\xi_1,...,\xi_m)=\varkappa(\prod_{i=1}^{m}\xi_i)-\sum_{i=1}^{m}\varkappa(\xi_i), \quad \varkappa=\omega,\ \Omega.
\end{equation}
\end{The}
\begin{proof}
 Since $g$ is a symmetry of equation \eqref{eqsym}, we have $[f, g]=0$, which is equivalent to $\pi_{k}([f,g])=0$, i.e.,
\begin{equation}\label{grasym}
\hat{\pi}_k(\phi([f,g]))=0,\,\,k=1,2,\ldots
\end{equation}
in symbolic representation. When $k=1$, it is zero for any choices of $\omega(\xi_1)$ and $\Omega(\xi_1)$. When $k=2$, vanishing of $\hat{\pi}_2\big(\phi([f,g])\big)
=[\omega(\xi_1), A_2(\xi_1, \xi_2)]+[a_2(\xi_1,\xi_2), \Omega(\xi_1)]$ is equivalent to
$$
\hu^2\bigg(A_2(\xi_1,\xi_2)\big(\omega(\xi_1)+\omega(\xi_2)- \omega(\xi_1\xi_2)\big)-a_2(\xi_1,\xi_2)\big(\Omega(\xi_1)+\Omega(\xi_2))-\Omega(\xi_1\xi_2)\big)\bigg)=0
$$
using \eqref{brom}. Taking into account the notation defined by \eqref{Gw}, we obtain the expression of $A_2(\xi_1,\xi_2)$ as in (\ref{symcoef2}).
Vanishing of $\hat{\pi}_k(\phi([f,g]))$
for $k\geq 3$ is equivalent to the expression of $A_k(\xi_1,\ldots,\xi_k)$ as in (\ref{symcoefk}).
\end{proof}
From Theorem \ref{thesym} it follows that a symmetry $g$ is uniquely determined by its linear part $\Omega(\xi_1)$. For a given $\Omega(\xi_1)$, all subsequent coefficients $A_k(\xi_1,\ldots,\xi_k)$ can be found recursively by formulae (\ref{symcoef2}) and (\ref{symcoefk}) in terms of the symbolic representation of the equation. However, it does not imply that every evolutionary equation possesses symmetries. To represent an element of the difference algebra $\A$ coefficients $A_k(\xi_1,\ldots,\xi_k)$ must be Laurent polynomials. For a generic $\Omega(\xi_1)$, coefficients $A_k(\xi_1,\ldots,\xi_k)$ given by (\ref{symcoef2})-(\ref{symcoefk}) are rational functions as these expressions contain denominators $G^{\omega}$ (unless $\omega(\xi_1)=const$ -- see Example \ref{Lin} below). In order to define symbols of elements of the difference algebra $\A$, these denominators must cancel with appropriate factors in the numerators.

We call a linear term $\hu\Omega(\xi_1),\,\Omega(\xi_1)\in\C[\xi_1,\xi_1^{-1}]$ {\it admissible} for the equation (\ref{eqsym}) if it is a linear term of a symmetry of the equation, i.e., all the coefficients $A_k$ given by (\ref{symcoef2})-(\ref{symcoefk}) are Laurent polynomials. As a $\C$-linear combination of symmetries is again a symmetry, the set of all admissible linear terms is a vector space over $\C$. We denote this space by $V_f$. For an integrable equation, the algebra of its symmetries and the vector space $V_f$ are infinite dimensional.

We say that $g$ is an approximate symmetry of degree $s$ if coefficients $A_k(\xi_1,\ldots,\xi_k)$, $k=2,\ldots,s$ are Laurent polynomials. When $f\in\overline{\A}$ is a formal series,  $g\in\overline{\A}$ with the symbolic representation (\ref{symsym}) is a symmetry if its coefficients $A_k$ are Laurent polynomials for arbitrary large $k$.
\begin{Ex} Consider the nonabelian Volterra equation \eqref{volt}.
\iffalse
\begin{equation}
\label{votlsymbsym}
u_t=uu_1-u_{-1}u.
\end{equation}
\fi
Note that the right hand side of the equation does not contain a linear term.  However, we introduce a shift transformation
\begin{equation}\label{shift}
u_i\to u_i+1,\,\,i\in\Z,
\end{equation}
which leads to
$$
u_t=f:=u_1-u_{-1}+uu_1-u_{-1}u,
$$
whose symbolic representation is
$$
\phi(f)=\hu\omega(\xi_1)+\hu^2a_2(\xi_1,\xi_2),\qquad
\omega(\xi_1)=\xi_1-\xi_1^{-1}, \quad a_2(\xi_1,\xi_2)=\xi_2-\xi_1^{-1}.
$$
It is known (see e.g. \cite{CasatiWang}) that the equation possesses an infinite hierarchy of symmetries with linear terms $\hu\Omega(\xi_1)=\hu (\xi_1^k-\xi_1^{-k}),\,\,k=2,3,\ldots$, i.e. $\hu (\xi_1^k-\xi_1^{-k})\in V_f,\,\,k=2,3,\ldots$.

In the case $k=2$, from the Theorem \ref{thesym} it follows that
$$
A_2(\xi_1,\xi_2)=\frac{(1+\xi_1)(1+\xi_2)(\xi_1^2\xi_2^2-1)}{\xi_1^2\xi_2},\quad A_3(\xi_1,\xi_2,\xi_3)=\frac{(\xi_1\xi_2\xi_3-1)(\xi_1\xi_2\xi_3+\xi_1\xi_2+\xi_1+1)}{\xi_1^2\xi_2},
$$
and $A_k(\xi_1,\ldots,\xi_k)=0,\,\,k>3$. This corresponds to
\begin{eqnarray*}
g&=&u_2-u_{-2}+(u+u_1)u_2-u_{-2}(u+u_{-1})+(u+u_1)u_1-u_{-1}(u+u_{-1})\\
&&+uu_1u_2-u_{-2}u_{-1}u+u(u+u_1)u_1-u_{-1}(u+u_{-1})u.
\end{eqnarray*}
By the inverse transformation of \eqref{shift}, i.e., $u_i\to u_i-1$, we obtain the symmetry of (\ref{volt}):
$$
u_{\tau}=uu_1u_2-u_{-2}u_{-1}u+u(u+u_1)u_1-u_{-1}(u+u_{-1})u-4(uu_1-u_{-1}u).
$$
%Note that the last term is the (scaled) Volterra equation itself.
As the equation itself is a trivial symmetry, we can remove the last term (the linear combination of symmetries is again a symmetry) and  obtain the same symmetry as in the Example \ref{EG2}.
\end{Ex}

\begin{Ex} Consider the nonabelian additive Bogoyavlensky equation
$$
u_t=u\left(\sum_{i=1}^nu_i\right)-\left(\sum_{i=1}^nu_{-i}\right)u,\quad n\in\N,\,n>1.
$$
We introduce a linear term by the transformation \eqref{shift} and consider instead
\begin{equation}
\label{bgexshift}
u_t=\sum_{i=1}^nu_i-\sum_{i=1}^nu_{-i}+u\left(\sum_{i=1}^nu_i\right)-\left(\sum_{i=1}^nu_{-i}\right)u.
\end{equation}
The expression for symbolic representation of the linear term $\hu\omega(\xi_1)$ is 
$$
\omega(\xi_1)=P(\xi_1)-P(\xi_1^{-1}),\quad P(\xi_1)=\xi_1^n+\xi_1^{n-1}+\cdots+1=\frac{\xi_1^{n+1}-1}{\xi_1-1}.
$$
It is known that equation (\ref{bgexshift}) possesses symmetries with linear terms (in the symbolic representation) given by
$\hu\Omega(\xi_1)=\hu\left(\big(P(\xi_1)\big)^k-\big(P(\xi_1^{-1})\big)^{k}\right),\,\,k=2,3,\ldots$ \cite{BG}. Theorem \ref{thesym} can be used to find the exact symbolic representations of symmetries for any fixed $k=2,3,\ldots$.
\iffalse
The above consideration also applies to the nonabelian multiplicative Bogoyavlensky equation
$$
u_t=uu_1\cdots u_n-u_{-n}u_{-n+1}\cdots u,\quad n\in\N,\,n>1.
$$
\fi
\end{Ex}

\begin{Ex}\label{Lin} Consider the following nonabelian D$\Delta$E:
$$
u_t=\alpha u+\sum_{k\ge 2}f_k,\quad f_k\in\A_k,\quad \alpha\in\C^*,
$$
whose linear term in the symbolic representation is $\alpha\hu$, i.e., $\omega(\xi_1)=\alpha=const$. This leads to the denominators $G^{\omega}(\xi_1,\ldots,\xi_k)=-(k-1)\alpha=const$ in (\ref{symcoef2}) and (\ref{symcoefk}). Therefore, this equation possesses symmetries with any choice of a linear term $\hu\Omega(\xi_1)$.
In fact, this equation is linearisable by an explicit changes of variables.

Let $\hat{f_k}=\hu^ka_k(\xi_1,\ldots,\xi_k)$, $k=2,3,\ldots$. Then $w=u+\sum_{k\ge 2}W_k,\,W_k\in\A_k$ satisfies a linear equation
$$
w_t=\alpha w,
$$
where the symbolic representations of $W_k$, i.e., $\phi(W_k)=\hu^k b_k(\xi_1,\ldots,\xi_k)$ are determined by

$$
b_2(\xi_1,\xi_2)=-\alpha^{-1}\hu^2a_2(\xi_1,\xi_2),
$$
and
\begin{eqnarray*}
&& b_p(\xi_1,\ldots,\xi_p)=\frac{1}{\alpha(1-p)} a_p(\xi_1,\ldots,\xi_p)\\
&&\quad +\frac{1}{\alpha(1-p)}  \sum_{k+s-1=p}\sum_{i=1}^kb_k(\xi_1,\ldots,\xi_{i-1},\prod_{j=i}^{i+s-1}\xi_j,\xi_{i+s},\ldots,\xi_{k+s-1})a_s(\xi_i,\ldots,\xi_{i+s-1}),\,\, p>2 .
\end{eqnarray*}
\end{Ex}

Theorem \ref{thesym} serves as an integrability test for a given equation or a class of equations if we assume the existence of a symmetry with a specified linear term $\hu\Omega(\xi_1)$. In such cases, the integrability conditions are that $A_k(\xi_1,\ldots,\xi_k),\,k=2,3,\ldots$ are Laurent polynomials. However, if the assumption of $\Omega(\xi_1)$ is incorrect, the test is inconclusive. It is a challenging problem to identify the vector space of admissible linear terms. In the differential case this problem had been solved for scalar polynomial evolutionary equations \cite{mr99g:35058}, two-component systems of evolutionary equations \cite{bswsys, swsys,vdk3, gz},  and odd order non-evolutionary equations \cite{nw07}. This remains an open problem to describe the sets of admissible linear terms for the scalar abelian and nonabelian D$\Delta$Es. In the subsequent section, we will formulate the necessary integrability conditions in the universal form, independent on the structure of the vector space of admissible linear terms.

\section{The symmetry approach and integrability conditions}\label{intconds}
The goal of the symmetry approach is to find necessary conditions for the existence of an infinite-dimensional algebra of symmetries for a given equation. This formalism was originally proposed and subsequently developed by Shabat and his team for the case of abelian PDEs (see e.g., \cite{SokShab,MikShabYam, MikShabSok}). The integrability conditions were formulated in terms of a formal recursion operator. Namely, the existence of an infinite-dimensional algebra of symmetries for an equation implies the existence of a formal pseudo-differential series of an arbitrary order, particularly of order 1, with coefficients belonging to an appropriate differential field/ring or its extension. This approach allows us to obtain necessary integrability conditions explicitly, which can be used not only as an integrability test for a given equation but also to solve classification problems by producing complete lists of integrable PDEs.
Later, the formalism was extended to integro-differential equations \cite{mn1, mn2}, some non-evolutionary PDEs \cite{mnw07}, and partial-difference and D$\Delta$Es \cite{Yam, adler14, MWX, MNWDiff}.

The universality of these integrability conditions means that the existence of the formal recursion operator follows from the existence of an infinite-dimensional algebra of symmetries and does not depend on the structure of the symmetry algebra. Deriving these conditions requires the existence of fractional powers of formal series. It is known that in both abelian and nonabelian cases of formal difference series, their fractional powers with coefficients in $\A$ or $\cM$ correspondingly may not exist. To address this issue in the abelian case, the quasi-local extension of the difference algebra was introduced \cite{MNWDiff}. In this section, we first provide this extension in the nonabelian case and subsequently formulate the integrability conditions in universal form, i.e., independent of the structure of the symmetry algebra.

\subsection{Quasi-local extension of the algebra of formal series $\hat{\A}((\eta))$}\label{ext}

In Section \ref{diffopsymb}, we have introduced the symbolic representation $\hat{\A}((\eta))$ for the algebra of formal series $\A((S))$. A generic element $\hat{B}$ in $\hat{\A}((\eta))$ is given by \eqref{fs}.
\iffalse
\begin{equation}
\label{fs}
\hat{A}=\sum_{p,q\ge 0}\hu_l^p\hu_r^qa_{pq}(\xi_1,\ldots,\xi_p,\eta,\zeta_1,\ldots,\zeta_q),
\end{equation}
where coefficients $a_{pq}(\xi_1,\ldots,\xi_p,\eta,\zeta_1,\ldots,\zeta_q)$ can be written as
\begin{eqnarray*}
&&a_{pq}(\xi_1,\ldots,\xi_p,\eta,\zeta_1,\ldots,\zeta_q)=\sum_{k\le N_{pq}}a_{pqk}(\xi_1,\ldots,\xi_p,\zeta_1,\ldots,\zeta_q)\eta^k,\quad N_{pq}\in\Z,
\end{eqnarray*}
and $a_{pqk}$ are Laurent polynomials:  $a_{pqk}(\xi_1,\ldots,\xi_p,\zeta_1,\ldots,\zeta_q)\in\C[\xi_1^{\pm 1},\ldots,\xi_p^{\pm 1},\zeta_1^{\pm 1},\ldots,\zeta_q^{\pm 1}]$.
\fi
In what follows it is convenient to consider more general formal series.

Let $\ring$ be the set of formal series
whose coefficients are {\it rational functions} in their variables, that is,
$$
\ring=\left\{\sum_{p,q\ge 0}\hu_l^p\hu_r^qa_{pq}(\xi_1,\ldots,\xi_p,\eta,\zeta_1,\ldots,\zeta_q)\,\big|\,a_{pq}\in\C(\xi_1,\ldots,\xi_p,\zeta_1,\ldots,\zeta_q,\eta)\right\}.
$$
This set is  a non-commutative algebra over $\C$ if we equip it with the natural addition, multiplication by constants and the composition rule (\ref{comp}). Clearly, we have $\hat{\A}((\eta))\subset \ring$, while the converse is not true, i.e., a generic element of $\ring$ is not a symbolic representation of a formal series in $\A((S))$. The algebra $\ring$ admits the same natural grading $\ring=\bigoplus_{p,q\ge 0}\ring_{p,q}$, where $\ring_{p,q}$ consists of elements of the form $\hu_l^p\hu_r^qa_{pq}(\xi_1,\ldots,\xi_p,\eta,\zeta_1,\ldots,\zeta_q)$.

\begin{Def} An expression $\hu_l^n\hu_r^ma(\xi_1,\ldots,\xi_n,\eta,\zeta_1,\ldots,\zeta_m)\in\ring$ is called $k$-local if its first $k$ coefficients
of its power expansion at $\eta\to\infty$
$$a(\xi_1,\ldots,\xi_n,\eta,\zeta_1,\ldots,\zeta_m)=\sum_{s\le p_{nm}}a^{(s)}(\xi_1,\ldots,\xi_n,\zeta_1,\ldots,\zeta_m)\eta^s, \quad p_{nm}\in\Z$$
are Laurent polynomials in all variables $\xi_1, \ldots, \xi_n, \zeta_1, \ldots, \zeta_m$.

A formal series $\hat{A}=\sum_{p,q\ge 0}\hu_l^p\hu_r^qa_{pq}(\xi_1,\ldots,\xi_p,\eta,\zeta_1,\ldots,\zeta_q)\in\ring$ is called $k$-local
if all its terms are $k$-local. We call the formal series $\hat{A}$ local if it is $k$-local for all $k\in \N$.
\end{Def}
Clearly, every formal series in $\hat{\A}((\eta))$ is local.

We now address the problem of extracting roots of a formal series.
Let $\hat{A}\in\ring$ be a formal series of the form 
\begin{equation}
\label{anroot}
\hat{A}=\eta^N+\sum_{p+q\ge 1}\hu_l^p\hu_r^qa_{pq}(\xi_1,\ldots,\xi_p,\eta,\zeta_1,\ldots,\zeta_q),\quad N\in\N.
\end{equation}
We formally seek a formal series ($N$-th root)
\begin{equation}
\label{broot}
\hat{B}=\eta+\sum_{p+q\ge 1}\hu_l^p\hu_r^qb_{pq}(\xi_1,\ldots,\xi_p,\eta,\zeta_1,\ldots,\zeta_q)\in \ring
\end{equation}
such that $\hat{B}^N=\hat{A}$. Using the composition rule (\ref{comp}) to compute $\hat{B}^N$ and equating the projections
$\hat{\pi}_{p,q}(\hat{B}^N-\hat{A})=0,\,p+q>0$, we find
\begin{eqnarray}
\label{root10}
p=1,q=0&:&\eta^{N-1}\Theta_N(\xi_1)b_{10}(\xi_1,\eta)-a_{10}(\xi_1,\eta)=0,\\
\label{root01}p=0,q=1&:& \eta^{N-1}\Theta_N(\zeta_1)b_{01}(\eta,\zeta_1)-a_{01}(\eta,\zeta_1)=0,\\
\label{rootpq}
p+q>1&:&\eta^{N-1}\Theta_N(\xi_1\cdots\xi_p\zeta_1\cdots\zeta_q)b_{pq}(\xi_1,\ldots,\xi_p,\eta,\zeta_1,\ldots,\zeta_q)+f_{pq}\\ \nonumber &&\quad\quad-a_{pq}(\xi_1,\ldots,\xi_p,\eta,\zeta_1,\ldots,\zeta_q)=0,
\end{eqnarray}
where the function $\Theta_N$ is given by
\begin{equation}
\label{thetan}
\Theta_N(\xi)=1+\xi+\cdots+\xi^{N-1}=\frac{\xi^N-1}{\xi-1}
\end{equation} 
and the functions $f_{pq}$ are completely determined by $b_{p'q'},\,p'+q'<p+q$.

Relations (\ref{root10})-(\ref{rootpq}) form a triangular system of equations for functions $b_{pq}$, which can be found successively:
\begin{eqnarray}
&&b_{10}(\xi_1,\eta)=\frac{a_{10}(\xi_1,\eta)}{\Theta_N(\xi_1)}\eta^{1-N};\quad b_{01}(\eta,\zeta_1)=\frac{a_{01}(\eta,\zeta_1)}{\Theta_N(\zeta_1)}\eta^{1-N};\label{bcoef1}\\
&&b_{pq}(\xi_1,\ldots,\xi_p,\eta,\zeta_1,\ldots,\zeta_q)=\frac{a_{pq}(\xi_1,\ldots,\xi_p,\eta,\zeta_1,\ldots,\zeta_q)-f_{pq}}{\Theta_N(\xi_1\cdots\xi_p\zeta_1\cdots\zeta_q)}\eta^{1-N},
\label{bcoefpq}
\end{eqnarray}
where $p+q>1$.
Therefore, there exists a formal series $\hat{B}\in\ring$ such that $\hat{B}^N=\hat{A}$.

Assume that the formal series $\hat{A}$ is local. Then from the relations (\ref{bcoef1})-(\ref{bcoefpq}), it follows that for a generic $\hat{A}$, the formal series $\hat{B}$ is not local as the denominators of $b_{pq}$ given by  (\ref{bcoef1})-(\ref{bcoefpq}) contain Laurent polynomials $\Theta_N(\xi_1\cdots\xi_p\zeta_1\cdots\zeta_q)$: the locality occurs if and only if $a_{pq}(\xi_1,\ldots,\xi_p,\eta,\zeta_1,\ldots,\zeta_q)-f_{pq}$ are divisible by $\Theta_N(\xi_1\cdots\xi_p\zeta_1\cdots\zeta_q)$ for every $p+q>0$, where we take $f_{10}=f_{01}=0$. We therefore consider $\Theta_N$-{\it quasi-local} extension of the algebra of formal series $\hat{\A}((\eta))$.

Let $\theta_N=\Theta^{-1}_N(\eta)$, where $\Theta_N(\eta)$ is given by (\ref{thetan}). The action of a pseudo-difference operator $\theta_N$ on an element $\hu^n_l\hu^m_r a(\xi_1,\ldots,\xi_n,\eta,\zeta_1,\ldots,\zeta_m)$ is defined as
$$
\theta_N(\hu^n_l\hu^m_r a(\xi_1,\ldots,\xi_n,\eta,\zeta_1,\ldots,\zeta_m))=\hu^n_l\hu^m_r\frac{a(\xi_1,\ldots,\xi_n,\eta,\zeta_1,\ldots,\zeta_m)}{\Theta_N(\xi_1\cdots\xi_n\zeta_1\cdots\zeta_m)}.
$$
We now define the sequence of extensions
$$
\check{\A}^{(N)}_{0}((\eta))=\hat{\A}((\eta)),\quad \check{\A}^{(N)}_{s+1}((\eta))=\overline{\check{\A}^{(N)}_{s}((\eta))\bigcup \theta_N(\check{\A}^{(N)}_{s}((\eta)))},\quad s=0,1,\ldots ,
$$ 
where the horizontal line denotes the closure under addition and multiplication operations. The index $s$ shows the maximal ``nesting" degree of the operator $\theta_N$. Clearly,
$$
\check{\A}^{(N)}_{0}((\eta))=\A((\eta))\subset \check{\A}^{(N)}_{1}((\eta))\subset\check{\A}^{(N)}_{2}((\eta))\subset\cdots,
$$
which depends on $N>1$ in the definition of the operator $\theta_N=\Theta^{-1}_N(\eta)$ defined as (\ref{thetan}).

\begin{Def} An element $\hu^n_l\hu^m_r a(\xi_1,\ldots,\xi_n,\eta,\zeta_1,\ldots,\zeta_m)\in\ring$ is called $k$-quasi-local if there exist $N>0$ and $s\ge 0$ such that the first $k$ terms of its power expansion at $\eta\to\infty$
\begin{eqnarray*}
\hu^n_l\hu^m_r a(\xi_1,\ldots,\xi_n,\eta,\zeta_1,\ldots,\zeta_m)=\sum_{i\le p_{nm}}\hu^n_l\hu^m_r a^{(i)}(\xi_1,\ldots,\xi_n,\zeta_1,\ldots,\zeta_m)\eta^i,
\quad p_{nm}\in \Z
\end{eqnarray*}
belong to $\check{\A}^{(N)}_{s}((\eta))$.

A formal series $\hat{A}=\sum_{p,q\ge 0}\hu_l^p\hu_r^qa_{pq}(\xi_1,\ldots,\xi_p,\eta,\zeta_1,\ldots,\zeta_q)$ is called $k$-quasi-local if all its terms are $k$-quasi-local. It is called quasi-local if it is $k$-quasi-local for every $k\in \N$.
\end{Def}

Note that $\Theta_N(\xi)|\Theta_{Nk}(\xi),\,k\in\N$. We have $\check{\A}^{(N)}_{s}((\eta))\subset \check{\A}^{(Nk)}_{s}((\eta))$ for $k\in\N,\ s\in \N$.
We say that an element/formal series is quasi-local with the {\it minimal extension} of order $N$ if $N$ is the minimal number such that the coefficients of the expansion belong to $\check{\A}^{(N)}_{s}((\eta))$ for some $s\in \N$.

From the computation above of the $N$-th root of a formal series (\ref{anroot}), it follows
\begin{Pro} Let 
$$
\hat{A}=\eta^N+\sum_{p+q\ge 1}\hu_l^p\hu_r^qa_{pq}(\xi_1,\ldots,\xi_p,\eta,\zeta_1,\ldots,\zeta_q),\quad N\in\N,
$$
be a formal series whose  terms $a_{pq},\,\,p+q\le s$ are local. Then there exists a unique formal series
$$
\hat{B}=\eta+\sum_{p+q\ge 1}\hu_l^p\hu_r^qb_{pq}(\xi_1,\ldots,\xi_p,\eta,\zeta_1,\ldots,\zeta_q),
$$
satisfying $\hat{B}^N=\hat{A}$, whose  terms $\hu_l^p\hu_r^qb_{pq},\,\,p+q\le s$ are quasi-local.
\end{Pro}
\begin{proof}
 From relations (\ref{bcoef1})-(\ref{bcoefpq}) follows that if $a_{pq},\,\,p+q\le s$ are local, for all $p+q\le s$ we have $\hu_l^p\hu_r^q b_{pq}\in\check{\A}^{(N)}_{s}((\eta))$.
\end{proof}

This proposition can naturally be extended to formal series beginning with $\h(\eta)\in\C[\eta,\eta^{-1}]$: polynomials generated by $\eta$ and $\eta^{-1}$ over $\C$.
\begin{Pro}\label{rootP} Let
$$
\hat{A}=\h(\eta)+\sum_{p+q\ge 1}\hu_l^p\hu_r^qa_{pq}(\xi_1,\ldots,\xi_p,\eta,\zeta_1,\ldots,\zeta_q), \quad \h(\eta)\in\C[\eta,\eta^{-1}]
$$
be a formal series whose  terms $a_{pq},\,\,p+q\le s$ are local and $\h(\eta)\ne const$. Then there exists a unique formal series
$$
\hat{B}=\eta+\sum_{p+q\ge 1}\hu_l^p\hu_r^qb_{pq}(\xi_1,\ldots,\xi_p,\eta,\zeta_1,\ldots,\zeta_q),
$$
whose  terms $b_{pq},\,\,p+q\le s$ are quasi-local, satisfying $\h(\hat{B})=\hat{A}$.
\end{Pro}

\subsection{Canonical formal recursion operator}\label{canfrop}
We proceed with deriving the necessary integrability conditions in their universal form.

Consider nonabelian D$\Delta$Es \eqref{ddn} of order $(-n, n)$ in the form
\begin{equation}
\label{eqcond}
u_t=f(u_{-n},\ldots,u_n)=\sum_{p\ge 1} f_p,\quad f_1\ne 0,\quad f_p\in\A_p,
\end{equation}
where $f$ is either a polynomial in $\A'$ or a formal series in $\overline{\A}'$.
The symbolic representation of equation (\ref{eqcond}) is presented as
\begin{equation}
\label{eqcondsymb}
\hu_t=\hat{f}=\hu\omega(\xi_1)+\sum_{p\ge 2}\hu^pa_p(\xi_1,\ldots,\xi_p),\quad a_p(\xi_1,\ldots,\xi_p)\in\C[\xi_1,\ldots,\xi_p].
\end{equation}
We assume that $\omega(\xi_1)\ne const$ (see Example \ref{Lin} for $\omega(\xi_1)=const$).

\begin{Def} A quasi-local formal series
\begin{equation}
\label{fropgen}
\Lambda= \h(\eta)+\sum_{p+q\ge 1}\hu_l^p\hu_r^q\h_{pq}(\xi_1,\ldots,\xi_p,\eta,\zeta_1,\ldots,\zeta_q),\quad \h(\eta)\in\C[\eta,\eta^{-1}]
\end{equation}
is called a formal recursion operator for equation (\ref{eqcondsymb}) if $\Lambda$ satisfies
\begin{equation}
\label{fropeq}
\Lambda_t-\hat{f}_*\circ\Lambda+\Lambda\circ\hat{f}_*=0,
\end{equation}
where $\Lambda_t=\D_{\hat f}(\Lambda)$ given by \eqref{symderser}.
\end{Def}
Equation (\ref{fropeq}) is linear in $\Lambda$, and for every $\h(\eta)\in\C[\eta,\eta^{-1}]$ and every $\hat{f}$, it possesses a unique formal solution $\Lambda\in\ring$. Namely, the following theorem holds:

\begin{The}\label{fropthe}
Let $\Lambda$ be a formal series of the form (\ref{fropgen}) satisfying (\ref{fropeq}). For every $\h(\eta)$ in $\C[\eta,\eta^{-1}]$,  the coefficients of the series can be recursively determined from the following system:
\begin{equation}
\label{frop1001}\h_{10}(\xi_1,\eta)=\frac{\h(\xi_1\eta)-\h(\eta)}{G^{\omega}(\xi_1,\eta)}a_2(\xi_1,\eta),\quad \h_{01}(\eta,\zeta_1)=\frac{\h(\zeta_1\eta)-\h(\eta)}{G^{\omega}(\zeta_1,\eta)}a_2(\eta,\zeta_1),
\end{equation}
\begin{eqnarray}
\label{froppq}&&\h_{pq}(\xi_1,\ldots,\xi_p,\eta,\zeta_1,\zeta_q) G^{\omega}(\xi_1,\ldots,\xi_p,\zeta_1,\ldots,\zeta_q,\eta)=\nonumber\\
&&=\left(\h(\xi_1\cdots\xi_p\zeta_1\cdots\zeta_q\eta)-\h(\eta)\right) a_{p+q+1}(\xi_1,\ldots,\xi_p,\eta,\zeta_1,\ldots,\zeta_q)\\ \nonumber
&&+\sum_{\substack{i+j\ge1,\\i'+j'\ge1}}^{\substack{i+i'=p,\\j+j'=q}}
\bigg(\h_{ij}(\xi_1,\ldots,\xi_i,\eta\prod_{s=i+1}^{i+i'}\xi_s\prod_{s'=1}^{j'}\zeta_{s'},\zeta_{j'+1},\ldots,\zeta_{j+j'})a_{i'+j'+1}(\xi_{i+1},\ldots,\xi_{i+i'},\eta,\zeta_1,\dots,\zeta_{j'})\\ \nonumber
&&\qquad-a_{i+j+1}(\xi_1,\ldots,\xi_i,\eta\prod_{s=i+1}^{i+i'}\xi_s\prod_{s'=1}^{j'}\zeta_{s'},\zeta_{j'+1},\ldots,\zeta_{j+j'})\h_{i'j'}(\xi_{i+1},\ldots,_{i+i'},\eta,\zeta_1,\ldots,\zeta_{j'})\bigg)\\ \nonumber
&&+\sum_{i=0}^{p-1}\sum_{k=1}^i\h_{iq}(\xi_1,\xi_{k-1},\prod_{s=k}^{k+p-i}\xi_s,\xi_{k+p-i+1},\ldots,\xi_p,\eta,\zeta_1,\ldots,\zeta_q)a_{p-i+1}(\xi_k,\ldots,\xi_{k+p-i})\\ \nonumber
&&+\sum_{j=0}^q\sum_{k=1}^j\h_{pj}(\xi_1,\ldots,\xi_p,\eta,\zeta_1,\ldots,\zeta_{k-1},\prod_{s=k}^{k+q-j}\zeta_s,\zeta_{k+q-j+1},\ldots,\zeta_q)a_{q-j+1}(\zeta_k,\ldots,\zeta_{k+q-j}),
\end{eqnarray}
where both $p$ and $q$ are nonnegative integers and $p+q>1$.
\end{The}
\begin{proof} To prove the statement it is sufficient to observe that relations
$$
\hat{\pi}_{pq}(\Lambda_t-\hat{f}_*\circ\Lambda+\Lambda\circ\hat{f}_*)=0
$$
are linear with respect to $\h_{pq}(\xi_1,\ldots,\xi_p,\eta,\zeta_1,\ldots,\zeta_q)$, and thus to
 resolve these recursively.
\end{proof}
 
Similar to Theorem \ref{thesym}, this theorem does not imply that every equation (\ref{eqcondsymb}) possesses a formal recursion operator since for a generic equation (\ref{eqcondsymb}), the coefficients of the formal series $\Lambda$ given by (\ref{frop1001})-(\ref{froppq}) are not necessarily quasi-local. In the next theorem we show that if the equation is integrable, i.e., it possesses an infinite hierarchy of symmetries, the obtained formal series must be quasi-local. Therefore, the conditions of quasi-locality of coefficients of formal series $\Lambda$ given by Theorem \ref{fropthe} are necessary integrability conditions for equation (\ref{eqcondsymb}).
 
\begin{The}\label{themain}
If equation (\ref{eqcondsymb}) is integrable, then a formal series
\begin{equation}
\label{fropcan}
\Lambda=\eta+\sum_{p+q\ge 1}\hu_l^p\hu_r^q\psi_{pq}(\xi_1,\ldots,\xi_p,\eta,\zeta_1,\ldots,\zeta_q)
\end{equation}
satisfying (\ref{fropeq}) is quasi-local.
\end{The}
The proof follows a similar structure to that of Theorem 4 in \cite{MNWDiff}. Nevertheless, we present the details here for the sake of completeness.
\begin{proof}
By Definition \ref{defsym}, the integrable equation (\ref{eqcondsymb}) possesses an infinite sequence of symmetries of increasing total orders.
Let 
$$
\hat{g}=\hu\Omega(\xi_1)+\sum_{p\ge 2}\hu^pA_p(\xi_1,\ldots,\xi_p)
$$
 be the symbolic representation of a symmetry $u_{\tau}=g$. Since $\omega(\xi_1)\ne 0$ we have by Proposition \ref{prolin} that the symmetries have non-vanishing linear part and thus  $\Omega(\xi_1)\ne 0$.
By Definition \ref{defsym1} we have $[\hat{g},\hat{f}]=0$. Therefore
\begin{equation}\label{fresym}
([\hat{g},\hat{f}])_*=\hat{g}_{*,t}+\hat{g}_*\circ\hat{f}_*-\hat{f}_{*,\tau}-\hat{f}_{*}\circ\hat{g}_*=0,
\end{equation}
where $\hat{g}_*$ is the Fr\'echet derivative of $\hat{g}$:
$$
\hat{g}_*=\Omega(\eta)+\sum_{p+q>0}\hu_l^p\hu_r^qA_{p+q+1}(\xi_1,\ldots,\xi_p,\eta,\zeta_1,\,\ldots,\zeta_q),
$$ 
$\hat{f}_*$ is the Fr\'echet derivative of $\hat{f}$:
$$
\hat{f}_*=\omega(\eta)+\sum_{p+q>0}\hu_l^p\hu_r^qa_{p+q+1}(\xi_1,\ldots,\xi_p,\eta,\zeta_1,\,\ldots,\zeta_q),
$$
and $\hat{f}_{*,\tau}$ denotes the action of the evolutionary derivation $\D_g$ on $\hat{f}_*$. We may rearrange \eqref{fresym} as
\begin{equation}
\label{mteq}
\hat{g}_{*,t}+\hat{g}_*\circ\hat{f}_*-\hat{f}_{*}\circ\hat{g}_*=\hat{f}_{*,\tau}.
\end{equation}
Denote by $\deg_+$ the highest power of $\eta$ of a Laurent formal series. We have $\deg_+(\hat{f}_*)=n$. Let $\deg_+(\Omega(\eta))=M$, so $\deg_+(\hat{g}_*)\ge M$, and $M$ can be chosen arbitrary large. We substitute $\hat{f}_*$, $\hat{g}_*$ in equation (\ref{mteq}) and notice that the degree of the right hand side of the equation is $\deg_+(\hat{f}_{*,\tau})\le n$, while $\deg_+(\hat{f}_*\circ\hat{g}_*)=\deg_+(\hat{g}_*\circ\hat{f}_*)\ge n+M$. Therefore, for every $p+q>0$ at least first $M-n$
terms of the power expansion of $\hat{g}_*$ at $\eta\to\infty$ in the form
$$\hu_l^p\hu_r^qA_{p+q+1}^{(k)}(\xi_1,\ldots,\xi_p,\zeta_1,\ldots,\zeta_q)\eta^k,\quad k=M,M-1,\ldots,n+1$$
 satisfy equation (\ref{fropeq}) with $\h(\eta)=\Omega(\eta)$. It follows from Theorem \ref{fropthe} that the solution of equation (\ref{fropeq}) exists and is unique. Thus we identify these terms with $\hu_l^p\hu_r^q\h_{pq}^{(k)}(\xi_1,\ldots,\xi_p,\zeta_1,\ldots,\zeta_q)\eta^k$ of the power expansion of $\hu_l^p\hu_r^q\h_{pq}(\xi_1,\ldots,\xi_p,\eta,\zeta_1,\ldots,\zeta_q)$ at $\eta\to\infty$. Since $\hat{g}_*$ is local, the obtained solution $\hat{\Lambda}$  of equation  (\ref{fropeq}) is $M$-local.

Consequently, as deduced from Proposition \ref{rootP}, there exists an $M$-quasi-local formal series (\ref{fropcan}).
Since $M$ can be chosen arbitrary large, this formal series is quasi-local.
\end{proof}

If $\Lambda$ in the form of (\ref{fropcan}) is a formal recursion operator for equation (\ref{eqcondsymb}), for any $\h(\eta)\in\C[\eta,\eta^{-1}]$, $\hat{\Lambda}=\h(\Lambda)$ is also a formal recursion operator of the equation starting with $\h(\eta)$. Coefficients of such a formal series can be determined using (\ref{frop1001}) and (\ref{froppq}). A constant is a trivial formal recursion operator. We call the formal recursion operator (\ref{fropcan}) starting with $\h(\eta)=\eta$ the {\it canonical} formal recursion operator.

Theorem \ref{themain} is constructive, and it provides necessary integrability conditions for equation (\ref{eqcondsymb}), independent on the structure of its algebra of symmetries. Moreover, it also provides some information of the symmetry structure of the equation. Namely, the following proposition holds:

\begin{Pro} Assume that the equation (\ref{eqcondsymb}) is integrable and the coefficients of its canonical formal recursion operator belong to the minimal extension $\check{\A}^{(N)}_{s}((\eta))$. Let $$g=\hu\Omega(\xi_1)+\sum_{p\ge 2}A_p(\xi_1,\ldots,\xi_p)$$
be a symmetry of the equation with $\deg_+(\Omega(\eta))=M>0$. Then $N\,|\,M$.
\end{Pro}
\begin{proof}
 The coefficients of a formal recursion operator obtained from $\hat{g}_*$ as described in the proof of Theorem \ref{themain} belong to extensions $\check{\A}^{(M)}_{s}((\eta))$. Therefore, the polynomial $\Theta_N(\eta)=\frac{\eta^{N}-1}{\eta-1}$ divides $\Theta_M(\eta)=\frac{\eta^{M}-1}{\eta-1}$ which occurs if and only if $N$ divides $M$.
\end{proof}
Theorems \ref{fropthe} and \ref{themain} show the advantage of symbolic representation. Explicit recurrence formulae (\ref{frop1001}) and (\ref{froppq}) are provided, which can be used to tackle the classification problem of integrable nonabelian D$\Delta$Es. For a given family of equations (\ref{eqcondsymb}), the process is as follows:
\begin{itemize}
\item find several first coefficients $\h_{pq}(\xi_1,\ldots,\xi_p,\eta,\zeta_1,\ldots,\zeta_q)$, starting from $p+q=1$;
\item Determine constraints on the equations imposed by the quasi-locality conditions of $\hu_l^p\hu_r^q\h_{pq}(\xi_1,\ldots,\xi_p,\eta,\zeta_1,\ldots,\zeta_q)$.
\end{itemize}
We illustrate the procedure by the following examples.

\begin{Ex} Let us describe all integrable equations of the form
\begin{equation}
\label{burgex}
u_t=u_1+\alpha u+c_1u_1u+c_2uu_1+c_3u^2,\quad \alpha,c_1,c_2,c_3\in\C.
\end{equation}
Its symbolic representation is of the form (\ref{eqcondsymb}) with
$$
\omega(\xi_1)=\xi_1+\alpha,\quad a_2(\xi_1,\xi_2)=c_1\xi_1+c_2\xi_2+c_3,
$$
and $a_p(\xi_1,\ldots,\xi_p)=0,\,\,p>2$.
The equation is of order $(0,1)$ with a total order of $1$.

By Theorems \ref{fropthe} and \ref{themain}, it follows that if this equation is integrable, its canonical formal recursion operator ($\h(\eta)=\eta$) must be local, as the terms of the canonical formal recursion operator must belong to subspaces $\check{\A}^{(1)}_s((\eta))$. However, $\theta_1=1$, so $\check{\A}^{(1)}_s((\eta))=\hat{\A}((\eta))$, where $s=1,2,\ldots$. Using formulae (\ref{frop1001}), we calculate the first two coefficients of the canonical formal recursion operator ($\h(\eta)=\eta)$:
$$
\h_{10}(\xi_1,\eta)=\frac{(\xi_1-1)(c_1\xi_1+c_2\eta+c_3)\eta}{\left((\xi_1-1)\eta-\alpha-\xi_1\right)},\quad \h_{01}(\eta,\zeta_1)=\frac{(\zeta_1-1)(c_1\eta+c_2\zeta_1+c_3)\eta}{\left((\zeta_1-1)\eta-\alpha-\zeta_1\right)}.
$$
The expansion of $\h_{10}(\xi_1,\eta)$ at $\eta\to\infty$ is of the form:
$$
\h_{10}(\xi_1,\eta)=c_2\eta+\sum_{i=1}^{\infty}\frac{(\xi_1+\alpha)^{i-1}(c_1\xi_1^2+(c_2+c_3-c_1)\xi_1+c_2\alpha-c_3)}{(\xi_1-1)^i}\eta^{i-1},
$$
and a similar expansion holds for $\h_{01}(\eta,\zeta_1)$.
The presence of powers of $\xi_1-1$ in denominators is obstruction to the locality of $\hu_l \h_{10}(\xi_1,\eta)$.
The cancellation of these denominators occurs if and only if  $c_1=c_2=c_3=0$ or $\alpha=-1$. In the former case the equation is linear.

Set $\alpha=-1$. Then
$$
\h_{10}(\xi_1,\eta)=\frac{c_2\eta+c_1\xi_1+c_3}{\eta-1}\eta,\quad \h_{01}(\eta,\zeta_1)=\frac{c_1\eta+c_2\zeta_1+c_3}{\eta-1}\eta,
$$
and thus $\hu_l\h_{10}(\xi_1,\eta)$, $\hu_r\h_{01}(\eta,\zeta_1)$ are local.

We now consider the locality conditions  of $\hu_l^2\h_{20}(\xi_1,\xi_2,\eta),\hu_l\hu_r\h_{11}(\xi_1,\eta,\zeta_1),\hu_r^2\h_{02}(\eta,\zeta_1,\zeta_2)$. The first term of the expansion of $\h_{11}(\xi_1,\xi_2,\eta)$ at $\eta\to\infty$ is of the form
$$
\h_{11}(\xi_1,\eta,\zeta_1)=-\frac{c_1c_2(\xi_1+\zeta_1-2)}{\xi_1\zeta_1-1}+O(\eta^{-1}).
$$
The presence of $\xi_1\zeta_1-1$ is the obstruction to locality of $\hu_l\hu_r\h_{11}(\xi_1,\eta,\zeta_1)$. Consequently, as $(\xi_1+\zeta_1-2)$ is not divisible by $\xi_1\zeta_1-1$, it follows that $c_1c_2=0$. Thus, we examine two cases:
$c_1\ne 0,\,c_2=0$ or $c_1=0,\,c_2\ne 0$.

If $c_1\ne 0$, then the subsequent terms in the expansion of $\h_{11}(\xi_1,\eta,\zeta_1)$ are
$$
\h_{11}(\xi_1,\eta,\zeta_1)=-\frac{c_1(c_3+c_1\xi_1)}{\zeta_1}\eta^{-1}-\frac{(c_1\xi_1+c_3)P(\xi_1,\zeta_1,c_1,c_3)}{\xi_1\zeta_1^2(\xi_1\zeta_1-1)}\eta^{-2}+O(\eta^{-3}),
$$
where
$$
P(\xi_1,\zeta_1,c_1,c_3)=c_1\xi_1^2 \zeta_1^2 + 2(c_1+c_3) \xi_1 \zeta_1^2 + c_1 \xi_1^2 \zeta_1-(c_1+c_3) \zeta_1^2 -(2 c_1+c_3) \xi_1 \zeta_1-c_1 \xi_1.
$$
We observe that the polynomial $\xi_1\zeta_1-1$ must divide $P(\xi_1,\zeta_1,c_1,c_3)$ for the locality of $\h_{11}(\xi_1,\eta,\zeta_1)$. As $P(\xi_1,\xi_1^{-1},c_1, c_3)=-(c_1+c_3)(\xi_1-1)^2$, it follows that $c_3=-c_1$. Consequently, we have
$$
\h_{11}(\xi_1,\eta,\zeta_1)=-\frac{c_1^2(\xi_1-1)\eta}{(\eta-1)(\zeta_1\eta-1)},\quad \h_{20}(\xi_1,\xi_2,\eta)=\frac{c_1^2(\xi_1-1)\xi_2\eta}{(\eta-1)(\xi_2\eta-1)},\quad \h_{02}(\eta,\zeta_1,\zeta_2)=0.
$$
It can be easily verified that $\hu_l^2\h_{20}(\xi_1,\xi_2,\eta),\hu_l\hu_r\h_{11}(\xi_1,\eta,\zeta_1),\hu_r^2\h_{02}(\eta,\zeta_1,\zeta_2)$ are all local. This leads to the equation
$$
u_t=u_1-u+c_1(u_1u-u^2).
$$
Upon rescaling $u_i\to \frac{1}{c_1}u_i,,i\in\Z$, we obtain $u_t=(u_1-u)(u+1)$. By applying the shift transformation $u_i\to u_i-1$, we arrive at one of the two nonabelian differential-difference Burgers equations
\begin{equation}
\label{burg1}
u_t=(u_1-u)u.
\end{equation}
When $c_2 \neq 0$, upon conducting analogous computations, we arrive at the other version of the nonabelian differential-difference Burgers equation:
\begin{equation}
\label{burg2}
u_t=u(u_1-u).
\end{equation}
We have thus demonstrated that the linear equation $u_t=u_1+\alpha u$, alongside equations (\ref{burg1}) and (\ref{burg2}), are the only integrable equations of the form (\ref{burgex}). Equation (\ref{burg1}) is linearisable by the transformation $u=v_1v^{-1}$, and $v$ satisfies $v_t=v_1-v$; likewise, equation (\ref{burg2}) is linearisable by the transformation $u=v^{-1}v_1$, yielding the same linear equation on $v$ (here we assume that the generators $v_i, , i \in \mathbb{Z}$ are invertible in the corresponding difference algebra generated by $v_i$, satisfying $v_iv_i^{-1}=v_i^{-1}v_i=1$).
\end{Ex}

\begin{Ex} Let us describe now all integrable equations of the form
\begin{equation}
\label{ex2}
u_t=f=u_2+\alpha u_1-\alpha u_{-1}-u_{-2}+c_1u_2u+c_2u_1u-c_2uu_{-1}-c_1uu_{-2},\quad \alpha,c_1,c_2\in\C.
\end{equation}
Note that it is skew-symmetric and that its symbolic representation is of the form (\ref{eqcondsymb}) with
$$
\omega(\xi_1)=\xi_1^2+\alpha\xi_1-\alpha\xi_1^{-1}-\xi_1^{-2},\quad a_2(\xi_1,\xi_2)=c_1(\xi_1^2-\xi_2^{-2})+c_2(\xi_1-\xi_2^{-1})
$$
and $a_p(\xi_1,\ldots,\xi_p)=0,\,\,p>2$. If equation (\ref{ex2}) is integrable, then the coefficients of the canonical formal recursion operator must belong to $\check{\A}^{(2)}_s((\eta))$ since $\deg_+(\omega(\eta))=2$. Using (\ref{frop1001}) and (\ref{froppq}), it can be easily verified that for $p+q\le2$, the coefficients $\hu_l^p\hu_r^q\h_{pq}$ are quasi-local for any choice of constants $\alpha, c_1$ and $c_2$. For instance, the expression for $\h_{10}(\xi_1,\eta)$ is given by
$$
\h_{10}(\xi_1,\eta)=\frac{c_1+(c_1\xi_1+c_2)\eta}{(\eta-1)\left(\xi_1(1+\xi_1)\eta^2+(\xi_1^2+\alpha\xi_1+2\xi_1+1)\eta+1+\xi_1\right)}\xi_1^2\eta
$$
and it is easy to see that all the coefficients of expansion of $\hu_l\h_{10}(\xi_1,\eta)$ at $\eta\to\infty$ belong to $\check{\A}^{(2)}_s((\eta))$ for some $s\ge 0$.

The expression for $\h_{03}(\eta,\zeta_1,\zeta_2,\zeta_3)$ is of the form
$$
\h_{03}(\eta,\zeta_1,\zeta_2,\zeta_3)=\frac{\Phi_{03}(\eta,\zeta_1,\zeta_2,\zeta_3,\alpha,c_1,c_2)}{\Psi_{03}(\eta,\zeta_1,\zeta_2,\zeta_3,\alpha,c_1,c_2)},
$$
where functions $\Phi_{03}, \Psi_{03}$ are polynomials of their arguments, and the polynomial $\Psi_{03}$ contains the irreducible factor
\begin{eqnarray*}
&&(\zeta_1\zeta_2\zeta_3\eta)^2G^{\omega}(\zeta_1,\zeta_2,\zeta_3,\eta)=\left(\zeta _1^2 \zeta _2^2 \zeta _3^2-1\right) \left(\eta ^4 \zeta _1^2 \zeta _2^2 \zeta _3^2+1\right)\\
&&\quad+\alpha  \eta  \zeta _1 \zeta _2 \zeta _3 \left(\zeta _1 \zeta _2 \zeta _3-1\right) \left(\eta ^2 \zeta _1 \zeta _2 \zeta _3+1\right)
+\eta ^2 (-\alpha  \zeta _2^2 \zeta _3^2 \zeta _1^3-\alpha  \zeta _2^2 \zeta _3^3 \zeta _1^2-\alpha  \zeta _2^3 \zeta _3^2 \zeta _1^2\\&&\quad+\alpha  \zeta _2 \zeta _3^2 \zeta _1^2+\alpha  \zeta _2^2 \zeta _3 \zeta _1^2+\alpha  \zeta _2^2 \zeta _3^2 \zeta _1
-\zeta _2^2 \zeta _3^2 \zeta _1^4-\zeta _2^2 \zeta _3^4 \zeta _1^2+\zeta _2^2 \zeta _1^2-\zeta _2^4 \zeta _3^2 \zeta _1^2+\zeta _3^2 \zeta _1^2+\zeta _2^2 \zeta _3^2).
\end{eqnarray*}
The presence of this factor results in violation of quasi-locality of $\h_{03}(\eta,\zeta_1,\zeta_2,\zeta_3)$, unless it is canceled out by the numerator $\Phi_{03}$. The cancellation occurs if and only if $\alpha=c_2=0$ or $\alpha=1$ and $c_2=c_1$.

In the first case the resulting nonlinear equation is
\begin{equation}
\label{exsv}
u_t=u_2-u_{-2}+c_1(u_2u-uu_{-2}), \quad c_1\neq 0.
\end{equation}
After rescaling $u_i\to\frac{1}{c_1}u_i$, for $i \in \mathbb{Z}$, and applying the shift transformation $u_i \to u_i - 1$, this equation transforms into the stretched nonabelian Volterra equation
$$
u_t=u_2u-uu_{-2}.
$$

In the second equation the resulting nonlinear equation is
\begin{equation}
\label{exb}
u_t=u_2+u_1-u_{-1}-u_{-2}+c_1 (u_1+u_2)u- c_1 u(u_{-1}+u_{-2}),\quad c_1\neq 0.
\end{equation}
After rescaling $u_i\to\frac{1}{c_1}u_i$, for $i \in \mathbb{Z}$, and applying the shift transformation $u_i \to u_i - 1$, this equation transforms into the nonabelian Bogoyavlensky equation
$$
u_t=(u_1+u_2)u-u(u_{-1}+u_{-2}).
$$
\end{Ex}
 We make a remark on the {\it locality} conditions concerning the formal recursion operators of integrable equations. For integrable abelian D$\Delta$Es, in our previous work \cite{MNWDiff}, we conjectured the existence of $\h(\eta)\in C[\eta,\eta^{-1}]$, such that the corresponding formal recursion operator $\Lambda=\h(\eta)+\sum_{p>0}\hu^p\h_p(\xi_1,\ldots,\xi_p,\eta)$ is {\it local}. Extending this conjecture to the nonabelian case, we propose that for every integrable nonabelian D$\Delta$E, there exists $\h(\eta)\in C[\eta,\eta^{-1}]$, ensuring the locality of the formal recursion operator (\ref{fropgen}). Furthermore, we speculate that this locality property holds when $\h(\eta)=\omega(\eta)_{+}$.

For instance, for the stretched Volterra equation (\ref{exsv}), the formal recursion operator is local if we choose $\h(\eta)=\eta^2$, while for the Bogoyavlensky equation (\ref{exb}), the locality occurs if we choose $\h(\eta)=\eta^2+\eta+1$. However, the detailed investigation of this conjecture remains beyond the scope of this paper.
\section{Classification of integrable nonabelian D$\Delta$Es}\label{classres}
In this section we apply the integrability test developed in the previous chapter and present the classification results of integrable nonabelian D$\Delta$Es
\begin{equation}
\label{eqclass}
u_t=f(u_{-n}, u_{-n+1},\ldots, u_{n}),\quad f\in\A,\quad n=1,2,3.
\end{equation}
We impose the following conditions on the right hand side of equation (\ref{eqclass}):
\begin{enumerate}
\item[(i)] Non-zero linear term: $\pi_1(f)\ne 0$, and $\pi_1(f)$ depends on $u_n$;
\item[(ii)] Quasi-linearity: $\frac{d^2}{d\epsilon^2}f(\epsilon u_{-n}, u_{-n+1},\ldots,u_{n-1}, \epsilon u_{n})=0$;
\item[(iii)] Skew-symmetry: $\cT(f)=-f$.
\item[(iv)] No lower order symmetries.
\end{enumerate}
We note that condition (i) together with condition (iii) implies that the symbolic representation of the linear term $\pi_1(f)$ is of the form $\phi(\pi_1(f))=\hu\sum_{i=1}^n c_i(\xi_1^i-\xi_1^{-i})$ and $c_n\ne 0$. 

We call equations {\it equivalent} if they are related by the
transposition automorphism $T$ defined by \eqref{trans} and by invertible transformations such as translation and scaling, i.e.,
$$u_k\to\alpha u_{k}+\beta,\quad t\to\gamma t,\quad \alpha,\gamma\in\C^*,\quad \beta\in\C,\quad k\in\Z.$$
In the classification lists it is sufficient to present a single representative from each equivalence class. As integrable equations are members of infinite hierarchies of symmetries and each hierarchy has a seed - a member of a minimal possible order -- we only present the seeds removing equations possessing lower order symmetries.

The classification strategy we employ is as follows: first, we obtain the complete list of equations of the form (\ref{eqclass}) that satisfy necessary integrability conditions -- the quasi-locality conditions for the canonical formal recursion operator. Then, we establish the integrability of each equation obtained by either presenting a Lax representation or by relating the equation to a known integrable one through a difference substitution.

We formulate our classification results in three theorems corresponding to $n=1,2,3$ in \eqref{eqclass}, respectively. The proofs of integrability are provided immediately after the statements of the theorems. Subsequently, we outline the proof for the classification lists collectively, as the approach remains the same for all fixed $n$.

\begin{The} \label{class1} Every nonlinear integrable nonabelian D$\Delta$E
$$
u_t=f(u_{-1}, u, u_1),
$$
satisfying conditions {\rm (i)-(iii) ($n=1$)}, is equivalent to one of the following equations:
\begin{subequations}
\begin{align}
%\label{lin1}
%u_t&=&u_1-u_{-1},\\
\tag{$\vt{1}{1}$}\label{Volterra1}
u_t=&uu_1-u_{-1}u,\\
\tag{$\vt{2}{1}$}\label{mVolterra11}
u_t=&(u^2-\alpha^2)u_1-u_{-1}(u^2-\alpha^2),\quad\alpha\in\C,\\
\tag{$\vt{3}{1}$}\label{mVolterra12}
u_t=&uu_1(u+\alpha)-(u+\alpha)u_{-1}u,\quad\alpha\in\C.
\end{align}
\end{subequations}
\end{The}

{\it Proof of Integrability}.
\underline{Equation (\ref{Volterra1})} is the well-known nonabelian Volterra equation. We present here for completeness its Lax representation (see e.g. \cite{BG})
$L_t=[A, L]$ with
$$
L=uS+\lambda S^{-1},\quad A=-u-u_{-1}-\lambda S^{-2}.
$$

\underline{Equations (\ref{mVolterra11}) and (\ref{mVolterra12})} are the  modified nonabelian Volterra equations labeled as ${\rm mVL^1}$ and ${\rm mVL^2}$ in \cite{Adler1},
where the transformations among these three equations are also presented.
 (\ref{mVolterra11}) is related to the Volterra equation
$$
w=ww_1-w_{-1}w,
$$
by Miura transformations:
$$
w=(u+\alpha)(u_1-\alpha).
$$
A Lax representation for \underline{Equation (\ref{mVolterra12})} is
$$
L=(u+\alpha)S^{-1}+\lambda u S,\quad A=u u_1 (1+\lambda S^2).
$$
$\hfill\square$

\begin{The} \label{class2} Every integrable nonabelian D$\Delta$E
$$
u_t=f(u_{-2}, u_{-1}, u, u_1, u_2),
$$
satisfying conditions {\rm (i)-(iv) ($n=2$)}, is equivalent to one of the following equations:
\begin{subequations}
\begin{align}
%%%%%%%%%%%%%%%%%%%%  Equation 12  %%%%%%%%%%%%%%%%%%%%%%%
\tag{$\vt{1}{2}$}\label{Volterra2}
u_t=&uu_2-u_{-2}u,\\
\tag{$\vt{2}{2}$}\label{mVolterra21}
u_t=&(u^2-\alpha^2)u_2-u_{-2}(u^2-\alpha^2),\quad\alpha\in\C,\\
%%%%%%%%%%%%%%%%%%%%  Equation 72  %%%%%%%%%%%%%%%%%%%%%%%
\tag{$\vt{3}{2}$}\label{mVolterra22}
u_t=&u u_2(u+\alpha)-(u+\alpha)u_{-2}u,\quad\alpha\in\C,\\
%%%%%%%%%%%%%%%%%%%%  Equation 45A  %%%%%%%%%%%%%%%%%%%%%%%
\tag{$\vt{4}{2}$}\label{eq45A}
u_t=&u  u_1  u_{2}-u_{-2}  u_{-1}  u + u  (u_{-1}-u_1)  u,\\
%%%%%%%%%%%%%%%%%%%%  Equation 80A  %%%%%%%%%%%%%%%%%%%%%%%
\tag{$\vt{5}{2}$}\label{eq80A}
u_t=&(u_{-1} + u)  u_{2} - u_{-2}  (u_{1} + u) + u  u_{1} -
 u_{-1}  u,\\
 %%%%%%%%%%%%%%%%%%%%  Equation 81  %%%%%%%%%%%%%%%%%%%%%%%
\tag{$\vt{6}{2}$}\label{eq81}
u_t=&u  u_{1}  u_{2}  u_{-1}  u-
 u  u_{1}  u_{-2}  u_{-1}  u+\alpha (u  u_{1}  u_{2}-u_{-2}  u_{-1}  u + u  (u_{-1}-u_{1})  u),\\%%%%%%%%%%%%%%%%%%%%  Equation 80B  %%%%%%%%%%%%%%%%%%%%%%%
\tag{$\vt{7}{2}$}\label{eq80B}
u_t=&u  u_{-1}  u  u_{1}  u_{2}- u_{-2}  u_{-1}  u  u_{1}  u
  +\alpha (u  u_{1}  u_{2}-u_{-2}  u_{-1}  u + u  (u_{-1}-u_1)  u),\\
 %%%%%%%%%%%%%%%%%%%%  Equation 47  %%%%%%%%%%%%%%%%%%%%%%%
\tag{$\vpt{8}{2}$} \label{eq47}
u_t=&u  u_{1}  u_{2} - u_{-2}  u_{-1}  u +
 u  (u - u_{1})  u_{1} - u_{-1}  (u - u_{-1})  u,
\end{align}
\end{subequations}
\begin{subequations}
\begin{align}
%\\
%\\
%\\
%%%%%%%%%%%%%%%%%%%%  Equation 79A  %%%%%%%%%%%%%%%%%%%%%%%
\tag{$\vpt{9}{2}$}\label{eq79A}
u_t=&(u + u_{-1})  (u_{1} + u)  u_{2} -
  u_{-2}  (u_{-1} + u)  (u_{1} + u)+
  u  (u + u_{-1})  u_{1} - u_{-1}  (u + u_{1})  u\\ \nonumber & -
  u  (u_{1} - u_{-1})  u,\\
 %%%%%%%%%%%%%%%%%%%%  Equation 74  %%%%%%%%%%%%%%%%%%%%%%%
\tag{$\vpt{10}{2}$}\label{eq74}
u_t=&u(u_1  u_2  u_1 -
 u_{-1}  u_{-2}  u_{-1} -
 u_1  u  u_1 +
 u_{-1}  u  u_{-1})u,\\
%%%%%%%%%%%%%%%%%%%%  Equation 77  %%%%%%%%%%%%%%%%%%%%%%%
%\label{eq77}
%u_t&=&(u  u_{-1} - 2)  u  u_1  u_{2} -
% u_{-2}  u_{-1}  u  (u_1  u - 2) +
% 2 u  (u_1 - u_{-1})  u,\\
 %%%%%%%%%%%%%%%%%%%%  Equation 79B  %%%%%%%%%%%%%%%%%%%%%%%
%\label{eq79B}
%u_t&=&u  u_{-1}  u  u_{1}  u_{2}-u_{-2}  u_{-1}  u  u_{1}  u+ \alpha u  (u_{-1} + u_{1})  u_{2}-\alpha u_{-2}  %(u_{-1} + u_{1})
%   u   + \alpha^2 (u_{2}-u_{-2}),\\
%%%%%%%%%%%%%%%%%%%%  Equation 45B  %%%%%%%%%%%%%%%%%%%%%%%
\tag{$\vpt{11}{2}$}\label{eq45B}
u_t=&(u^2-\alpha^2) (u_1^2 -\alpha^2) u_2 -
  u_{-2}  (u_{-1}^2 - \alpha^2) (u^2 - \alpha^2)  - (u^2 - \alpha^2)   u_1 u  u_1\\ \nonumber
   & +
  u_{-1}  u  u_{-1} (u^2 - \alpha^2) +
  u  u_{-1}  (u^2 - \alpha^2) u_1 - u_{-1} (u^2 - \alpha^2)  u_1  u,\\
%\end{align}
%\end{subequations}
%\begin{subequations}
%\begin{align}
%%%%%%%%%%%%%%%%%%%%  Equation 30  %%%%%%%%%%%%%%%%%%%%%%%
\tag{$\tot{1}{2}$}\label{eq30}
u_t=&(u  u_{-1} + 1)  (u  u_{1} + 1)  u_{2} -
 u_{-2}  (u_{-1}  u + 1)  (u_{1}  u + 1),\\
%%%%%%%%%%%%%%%%%%%%  Equation 100  %%%%%%%%%%%%%%%%%%%%%%%
\tag{$\bt{1}{2}$}\label{eq100}
u_t=&u(u_1+u_2)-(u_{-1}+u_{-2})u,\\
%%%%%%%%%%%%%%%%%%%%  Equation 136  %%%%%%%%%%%%%%%%%%%%%%%
\tag{$\bt{2}{2}$}\label{eq136}
u_t=&uu_1u_2-u_{-2}u_{-1}u,\\
%%%%%%%%%%%%%%%%%%%%  Equation 132A  %%%%%%%%%%%%%%%%%%%%%%%
\tag{$\bt{3}{2}$}\label{eq132A}
u_t=&u  u_{1}  u_{2}-u_{-2}  u_{-1}  u+ u ^2  u_{1} - u_{-1}  u ^2 +
 u  (u_{-1}-u_{1})  u ,\\
  %%%%%%%%%%%%%%%%%%%%  Equation 132B  %%%%%%%%%%%%%%%%%%%%%%%
\tag{$\bt{4}{2}$}\label{eq132B}
u_t=&u  (u + \alpha)
  u_{1}  (u_{2} + \alpha) - (u_{-2} + \alpha)
  u_{-1}  (u + \alpha)  u\\ \nonumber
   & + (u + \alpha)  u_{-1}
  u  (u_{1} + \alpha) - (u_{-1} + \alpha)  u
  u_{1}  (u + \alpha),\\
  %%%%%%%%%%%%%%%%%%%%  Equation 144  %%%%%%%%%%%%%%%%%%%%%%%
\tag{$\bt{5}{2}$}\label{eq144}
u_t=&u  u_{1}  u_{2}  (u+\alpha ) - ( u+\alpha)
  u_{-2}  u_{-1}  u,\\
 %%%%%%%%%%%%%%%%%%%%  Equation 147A  %%%%%%%%%%%%%%%%%%%%%%%
\tag{$\bt{6}{2}$}\label{eq147A}
u_t=&(u+u_{-1})(u+u_1)(u_1+u_2)-(u_{-2}+u_{-1})(u+u_{-1})(u+u_1),\\
 %%%%%%%%%%%%%%%%%%%%  Equation 147  %%%%%%%%%%%%%%%%%%%%%%%
\tag{$\bt{7}{2}$}\label{eq147}
u_t=&(uu_{-1}+\alpha)(uu_1+\alpha)uu_1u_2-u_{-2}u_{-1}u(u_{-1}u+\alpha)(u_1u+\alpha),\\
%%%%%%%%%%%%%%%%%%%%  Equation 148  %%%%%%%%%%%%%%%%%%%%%%%
\tag{$\bt{8}{2}$}\label{eq148}
u_t=&(uu_{-1}+\alpha)(u_1u+\alpha)(u_2u_1+\alpha)u-u(u_{-1}u_{-2}+\alpha)(uu_{-1}+\alpha)(u_1u+\alpha).
\end{align}
\end{subequations}
\end{The}

{\it Proof of Integrability}.
The list of integrable equations of order $(-2,2)$ naturally splits into four sublists: equations of  the stretched Volterra type (List 1), labelled with letter ``V'';  equations related by nonautonomous transformations to symmetries of Volterra type chains of order $(-1,1)$ in Theorem \ref{class1}  (List 2), labelled with ``${\rm V'}$''; equations of the relativistic Toda type (List 3), labelled with ``T''; and equations of the Bogoyavlensky type (List 4), labelled with ``B''.

\begin{itemize}
\item[List 1:] \underline{Equations (\ref{Volterra2}), (\ref{mVolterra21}) and (\ref{mVolterra22})} are the stretched versions ($u_k\to u_{2k},\,k\in\Z$) of (\ref{Volterra1}), (\ref{mVolterra11}) and (\ref{mVolterra12}).

\underline{Equations (\ref{eq45A}), (\ref{eq80A})} are related to the stretched Volterra equation
$$
w_t=ww_2-w_{-2}w
$$
via a Miura transformations
$$
\mbox{(\ref{eq45A})}:\,w=uu_1,\quad \mbox{(\ref{eq80A})}:\, w=u+u_1.
$$

\underline{Equations (\ref{eq81}) and (\ref{eq80B})} are related to stretched Volterra equations
$$
w_t=\left(w-\frac{\alpha}{2}\right)w_2\left(w+\frac{\alpha}{2}\right)-
\left(w+\frac{\alpha}{2}\right)w_{-2}\left(w-\frac{\alpha}{2}\right)
$$
and
$$
w_t=\left(w^2-\frac{\alpha^2}{4}\right)w_2-\left(w^2-\frac{\alpha^2}{4}\right)w_{-2}
$$
correspondingly via a Miura transformation
$
w=\frac{\alpha}{2}+uu_1.
$

\item[List 2:] \underline{Equation (\ref{eq47})} is related to a symmetry of the Volterra equation (\ref{Volterra1}) via a  non-autonomous transformation $u_k\to (-1)^ku_k,\,k\in\Z$.

\underline{Equation (\ref{eq79A})} is related to equation (\ref{eq47})
$$
w_t=ww_1w_2-w_{-2}w_{-1}w+w(w-w_1)w_1-w_{-1}(w-w_{-1})w
$$
via a Miura transformation $w=u+u_1$.

\underline{Equation (\ref{eq74})} is related to a symmetry of the equation (\ref{mVolterra12}) with $\alpha=0$ via a non-autonomous transformation $u_k\to \i^ku_k,\,k\in\Z$.

\underline{Equation (\ref{eq45B})}: A symmetry of the non-autonomous equation $$u_{nt}=(u_n^2-\alpha^2 \i^{2n})u_{n+1}-u_{n-1}(u_n^2-\alpha^2 \i^{2n})$$ is
\begin{eqnarray*}
u_{nt}&=&(u_n ^2 - \alpha^2 \i^{2n})   (u_{n+1}^2 + \alpha^2 \i^{2n})   u_{n+2} -
 u_{n-2}   (u_{n-1}^2 + \alpha^2 \i^{2n})   (u_{
      n}^2 - \alpha^2 \i^{2n})\\
 &+& (u_n ^2 -  \alpha^2 \i^{2n})   u_{n+1}   u_n
    u_{n+1} - u_{n-1}   u_n    u_{n-1}   (u_n^2 -  \alpha^2 \i^{2n}) \\
    &+& u_n
    u_{n-1}   (u_n^2 - \alpha^2 \i^{2n})   u_{n+1} -
   u_{n-1}   (u_n^2 -  \alpha^2 \i^{2n})   u_{n+1}   u_n .
\end{eqnarray*}
Under the transformation $u_n\to u_n \i^n$, it becomes the equation (\ref{eq45B}).

\item[List 3:] \underline{Equation (\ref{eq30})} possesses the following Lax representation:
$$
L=P^{-1}Q,\quad P=(1-S^{-2})u_1^{-1},\quad Q=\left((uu_1+1)(u_2u_1+1)S^2-1\right)u_{-1}^{-1}
$$
$$
A=(u_1u+1)(u_1u_2+1)S^2-u_{-1}\left((uu_1+1)(u_2u_1+1)-1\right)u_1^{-1}-S^{-2}.
$$
(Here we assume that every element $u_k,\,k\in\Z$, is invertible, and we consider the equation and the corresponding structures on a factor algebra $\A/\langle u_ku_k^{-1}-1,\,u_{k}^{-1}u_k-1\rangle,\,k\in\Z$). The abelian version of this equation was found and studied in \cite{GMY14}.

\item[List 4:]  \underline{Equations (\ref{eq100}) and (\ref{eq136})} are additive and multiplicative versions of order $(-2,2)$ of the nonabelian Bogoyavlensky chains \cite{BG}:
\begin{eqnarray}
\label{BGaddn}
u_t&=&u(\sum_{k=1}^nu_k)- (\sum_{k=1}^nu_{-k})u,\\
\label{BGmultn}
u_t&=&uu_1\cdots u_n-u_{-n}\cdots u_{-1}u
\end{eqnarray}
For completeness we present here their Lax representations:
\begin{eqnarray*}
&\mbox{Equation (\ref{BGaddn})}&:\, L=uS^n+\lambda S^{-1},\quad A=-\sum_{k=1}^nu_{-k}-\lambda S^{-n-1},\\
&\mbox{Equation (\ref{BGmultn})}&:\, L=u S+\lambda S^{-n},\quad A=\lambda^{-1}uu_1\cdots u_nS^{n+1}.
\end{eqnarray*}

\underline{Equations (\ref{eq132A}) and (\ref{eq132B})} are related to the additive nonabelian Bogoyavlensky chain
$$
w_t=w(w_1+w_2)-(w_{-1}+w_{-2})w
$$
via Miura type transformations
$$
\mbox{(\ref{eq132A})}:\,w=uu_1,\quad \mbox{(\ref{eq132B})}:\, w=(u+\alpha)u_1u_2.
$$

\underline{Equation (\ref{eq144})} is a 2-relative of the family of the nonabelian Bogoyavlensky chain, whose 1-relative is \ref{mVolterra12}:
\begin{equation}
\label{BGmult2n}
u_t=uu_1\cdots u_n(u+\alpha)-(u+\alpha)u_{-n}\cdots u_1u.
\end{equation}
The corresponding Lax representation is
$$
L=(u+\alpha)S^{-n}+\lambda uS,\quad A=uu_1\cdots u_n(1+\lambda S^{n+1}).
$$
The reduction $\alpha=0$ of this equation appears in \cite{BG}.

\underline{Equations (\ref{eq147A}) and (\ref{eq147}) } are related to the multiplicative nonabelian Bogoyavlensky chain
$$
w_t=ww_1w_2-w_{-2}w_{-1}w
$$
via Miura transformations
$$
\mbox{(\ref{eq147A})}:\, w=u+u_1,\quad  \mbox{(\ref{eq147})}:\, w=(uu_1+\alpha)u_2.
$$

\underline{Equation (\ref{eq148})}: Miura type transformation $w=u_1u+\alpha$  relates this equation to
$$
w_t=ww_1w_2(w-\alpha)-(w-\alpha)w_{-2}w_{-1}w,
$$
which is equivalent upon $\alpha\to-\alpha$ to equation (\ref{eq144}). $\hfill \square$
\end{itemize}

\begin{The} \label{class3} Every integrable nonabelian D$\Delta$E
$$
u_t=f(u_{-3}, u_{-2}, u_{-1}, u, u_1, u_2, u_3),
$$
satisfying conditions {\rm (i)-(iv) ($n=3$)}, is equivalent to one of the following equations:
\begin{subequations}
\begin{align}
%%%%%%%%%%%%%%%%%%%%  Equation 748  %%%%%%%%%%%%%%%%%%%%%%%
\tag{$\vt{1}{3}$}\label{strV31}
u_t=&uu_3-u_{-3}u,\\
%%%%%%%%%%%%%%%%%%%%  Equation 739  %%%%%%%%%%%%%%%%%%%%%%%
\tag{$\vt{2}{3}$}\label{strV32}
u_t=&(u^2-\alpha^2)u_3-u_{-3}(u^2-\alpha^2),\\
\tag{$\vt{3}{3}$}\label{strV33}
u_t=&u u_3(u+\alpha)-(u+\alpha)u_{-3}u,\\
%%%%%%%%%%%%%%%%%%%%  Equation 743  %%%%%%%%%%%%%%%%%%%%%%%
\tag{$\vt{4}{3}$}\label{strV3}
u_t=&uu_{-1}u_1u_3-u_{-3}u_{-1}u_1u,\\
%%%%%%%%%%%%%%%%%%%%  Equation 702  %%%%%%%%%%%%%%%%%%%%%%%
\tag{$\vt{5}{3}$}\label{eq702V3}
u_t=&uu_1u_2u_3-u_{-3}u_{-2}u_{-1}u-u(u_1u_2-u_{-2}u_{-1})u,\\
%%%%%%%%%%%%%%%%%%%%  Equation 32A  %%%%%%%%%%%%%%%%%%%%%%%
\tag{$\vt{6}{3}$}\label{eq32AV3}
u_t=&(u + u_{-1} + u_{-2})u_3 - u_{-3} (u+ u_{1} + u_{2}) +
  u (u_1 + u_2) - (u_{-1} + u_{-2}) u,\\
%%%%%%%%%%%%%%%%%%%%  Equation 32 general  %%%%%%%%%%%%%%%%%%%%%%%
\tag{$\vt{7}{3}$}\label{eq32V3}
u_t=&(u   u_{-2}   u_{-1} + \alpha)   u   u_{1}   u_{2}   u_{3} -
 u_{-3}   u_{-2}   u_{-1}
  u   (u_{1}   u_{2}   u + \alpha) - \alpha u   (u_{1}   u_{2} - u_{-2}   u_{-1})   u,\\
%%%%%%%%%%%%%%%%%%%%  Equation 597 general  %%%%%%%%%%%%%%%%%%%%%%%
\tag{$\vt{8}{3}$}\label{eq597V3}
u_t=&u   u_{1}   u_{2}
   u_{3}   (u_{-2}   u_{-1}   u + \alpha) - (u   u_{1}   u_{2} + \alpha)
    u_{-3}   u_{-2}   u_{-1}   u -
  \alpha u(u_{1}   u_{2} - u_{-2}   u_{-1})u,\\
%\end{align}
%\end{subequations}
%\begin{subequations}
%\begin{align}
%%%%%%%%%%%%%%%%%%%%  Equation 747  %%%%%%%%%%%%%%%%%%%%%%%
\tag{$\bt{1}{3}$}\label{BGV3}
u_t=&u   (u_{1} + u_{2} + u_{3}) - (u_{-1} + u_{-2} + u_{-3})   u,\\
 %%%%%%%%%%%%%%%%%%%%  Equation 689  %%%%%%%%%%%%%%%%%%%%%%%
\tag{$\bt{2}{3}$}\label{BGVm3}
u_t=&uu_1u_2u_3-u_{-3}u_{-2}u_{-1}u,\\
 %%%%%%%%%%%%%%%%%%%%  Equation 701  %%%%%%%%%%%%%%%%%%%%%%%
\tag{$\bt{3}{3}$}\label{BGV23}
u_t=&u   u_{1}   u_{2}   u_{3} - u_{-3}   u_{-2}   u_{-1}   u +
 C_u\left( u u_{1}   u_{2}+u_{-1}   u   u_{1}+u_{-2}   u_{-1} u \right) ,\\
%%%%%%%%%%%%%%%%%%%%  Equation 744  %%%%%%%%%%%%%%%%%%%%%%%
\tag{$\bt{4}{3}$}\label{BGV13}
u_t=&u   u_{1}   u_{3} - u_{-3}   u_{-1}   u +
 u^2 u_{2} - u_{-2} u^2 +C_u \left(u_{-1}u_1\right),\\
 %%%%%%%%%%%%%%%%%%%%  Equation 687  %%%%%%%%%%%%%%%%%%%%%%%
\tag{$\bt{5}{3}$}\label{BGV33}
u_t=&(u + \alpha)   u   u _{1}   u _{2}   u _{3} -
 u _{-3}   u _{-2}   u _{-1}   u   (u + \alpha) +
 u   (u _{-1} + \alpha)   u   u _{1}   u _{2} - (u _{-1} + \alpha)   u
  u _{1}   u _{2}   u \\ \nonumber &-
 u _{-2}   u _{-1}   u   (u _{1} + \alpha)   u +
 u   u _{-2}   u _{-1}   u   (u _{1} + \alpha) +
 \alpha (u   u _{-1}   u   u _{1} - u _{-1}   u   u_{1}   u),\\
 %%%%%%%%%%%%%%%%%%%%  Equation 662  %%%%%%%%%%%%%%%%%%%%%%%
\tag{$\bt{6}{3}$}\label{BGV43}
u_t=&uu_1u_2u_3(u+\alpha)-(u+\alpha)u_{-3}u_{-2}u_{-1}u,\\
  %%%%%%%%%%%%%%%%%%%%  Equation 639  %%%%%%%%%%%%%%%%%%%%%%%
\tag{$\bt{7}{3}$}\label{BGV53}
u_t=&uu_{-1}u_1uu_2u_1u_3-u_{-3}u_{-1}u_{-2}uu_{-1}u_1u,\\
 %%%%%%%%%%%%%%%%%%%%  Equation 19B  %%%%%%%%%%%%%%%%%%%%%%%
\tag{$\bt{8}{3}$}\label{BGV63}
u_t=&(u_{-2}+u_{-1}+u)(u_{-1}+u+u_1)(u+u_1+u_2)(u_1+u_2+u_3)\\ \nonumber &-(u_{-3}+u_{-2}+u_{-1})(u_{-2}+u_{-1}+u)(u_{-1}+u+u_1)(u+u_1+u_2),\\
 %%%%%%%%%%%%%%%%%%%%  Equation 19  %%%%%%%%%%%%%%%%%%%%%%%
\tag{$\bt{9}{3}$}\label{BGV73}
u_t=&(uu_{-2}u_{-1}+\alpha)(uu_1u_{-1}+\alpha)(uu_1 u_2+\alpha)uu_1u_2u_3\\ \nonumber&-u_{-3}u_{-2}u_{-1}u(u_{-2}u_{-1}u+\alpha)(u_1u_{-1} u+\alpha)(u_1 u_2 u+\alpha),\\
%%%%%%%%%%%%%%%%%%%%  Equation 21  %%%%%%%%%%%%%%%%%%%%%%%
\tag{$\bt{10}{3}$}\label{BGV83}
u_t=&(uu_{-1}u_{-2}+\alpha)(u_1uu_{-1}+\alpha)(u_2u_{1}u+\alpha)(u_3u_{2}u_{1}+\alpha)u\\ \nonumber &-u(u_{-1}u_{-2}u_{-3}+\alpha)(uu_{-1}u_{-2}+\alpha)(u_1uu_{-1}+\alpha)(u_2u_{1}u+\alpha) ,
\end{align}
\end{subequations}
where $C_u=\cL_u-\cR_u$.
\end{The}
{\it Proof of Integrability}. The list of integrable equations of order $(-3,3)$ splits into two sublists: equations of the Volterra type (List 1) and equations of the Bogoyavlensky type (List 2).

\begin{itemize}
\item[List 1:] \underline{Equations (\ref{strV31}), (\ref{strV32}) and (\ref{strV33})} are the stretched versions ($u_k\to u_{3k},\,k\in\Z$) of (\ref{Volterra1}), (\ref{mVolterra11}) and (\ref{mVolterra12}).

\underline{Equations (\ref{strV3}), (\ref{eq702V3}), (\ref{eq32AV3}), (\ref{eq32V3})} are related to the stretched Volterra chain
$$
w_t=ww_3-w_{-3}w
$$
via Miura transformations
$$
\begin{array}{ll}
 \mbox{(\ref{strV3})}:\, w=u_{-1}u_1u_3,\quad & \mbox{(\ref{eq702V3})}:\, w=uu_1u_2,\\
 \mbox{(\ref{eq32AV3})}:\, w=u+u_1+u_2,\quad & \mbox{(\ref{eq32V3})}:\, w=u_{-3}u_{-2}u_{-1}(uu_1u_2+\alpha).
\end{array}
$$
\underline{Equation (\ref{eq597V3})}: transformation $w=-\frac{\alpha}{2}-u_1u_2u_3$ brings the equation (\ref{eq597V3}) to
$$
w_t=\left(w+\frac{\alpha}{2}\right)w_3 \left(w-\frac{\alpha}{2}\right)-\left(w-\frac{\alpha}{2}\right)w_{-3} \left(w+\frac{\alpha}{2}\right),
$$
which is equivalent to equation (\ref{strV33}).

\item[List 2:] \underline{Equations (\ref{BGV3}) and (\ref{BGVm3})} are additive and multiplicative versions of order $(-3,3)$ of the nonabelian Bogoyavlensky chains (\ref{BGaddn}) and (\ref{BGmultn}).

\underline{Equations (\ref{BGV23}), (\ref{BGV13}), (\ref{BGV33})} are related to the additive nonabelian Bogoyavlensky chain of order $(-3,3)$:
$$
w_t=w(w_1+w_2+w_3)-(w_{-1}+w_{-2}+w_{-3})w.
$$
via Miura transformations
$$
\mbox{(\ref{BGV23})}:\, w=uu_1u_2,\quad \mbox{(\ref{BGV13})}:\, w=uu_2,\quad \mbox{(\ref{BGV33})}:\, w=uu_1u_2(u_3+\alpha).
$$

\underline{Equation (\ref{BGV43})} is $3$-relative of a family of multiplicative nonabelian Bogoyavlensky chains (\ref{BGmult2n}).

\underline{Equations (\ref{BGV53}), (\ref{BGV63}), (\ref{BGV73})}
are related to the multiplicative nonabelian Bogoyavlensky chain of order $(-3,3)$:
$$
w_t=ww_1w_2w_3-w_{-3}w_{-2}w_{-1}w
$$
via Miura transformations
$$
\mbox{(\ref{BGV53})}:\, w=uu_2,\quad \mbox{(\ref{BGV63})}:\, w=u+u_1+u_2,\quad \mbox{(\ref{BGV73})}:\, w=(uu_1 u_2+\alpha)u_3.
$$

\underline{Equation (\ref{BGV83})}: Miura type transformation $w=u_2u_1 u+\alpha$ brings the equation (\ref{BGV83}) to
$$
w_t=w w_1 w_2 w_3 (w-\alpha)-(w-\alpha) w_{-3} w_{-2} w_{-1} w,
$$
which is equivalent upon $\alpha\to-\alpha$ to equation (\ref{BGV43}).
$\hfill\square$
\end{itemize}

%In the above lists the notation $[\cdot,\cdot]$ stays for the commutator $[a,b]:=ab-ba$.

The proof of the classification for  Theorems \ref{class1}-\ref{class3} relies on the following statement:

\begin{Pro}\label{prounique}
Let
$$
u_t=\sum_{p\ge 1}f_p,\quad f_i\in\A_i\quad\mbox{and}\quad u_t=\sum_{p\ge 1}\tilde{f}_p,\quad\tilde{f}_i\in\A_i
$$
be two integrable D$\Delta$Es, such that
\begin{itemize}
\item The linear terms are $f_1=\tilde{f}_1=u_n-u_{-n}+\sum_{k=1}^{n-1}\alpha_k (u_k-u_{-k}),\quad\alpha_k\in\C,\quad n=1,2,3$,
\item The quadratic and cubic terms coincide: $f_2=\tilde{f}_2,\,\,f_3=\tilde{f}_3$.
\end{itemize}
Then the equations coincide, i.e. $f_p=\tilde{f}_p,\quad \forall p\in\N$.
\end{Pro}
The proof of this proposition is fully analogous to the proof of Proposition 8 in \cite{MNWDiff}, and therefore we omit it.

{\it Sketch of the proof of classification for Theorems \ref{class1}-\ref{class3}}. For each $n=1,2,3$ the symbolic representation of a generic equation (\ref{eqclass}) satisfying conditions {\rm (i)-(iii)} is of the form (\ref{eqcondsymb}) with
\begin{eqnarray*}
&&\omega(\xi)=\sum_{k=1}^n\alpha_k(\xi^k-\xi^{-k}),\quad \alpha_n\ne 0,\\
&&a_p(\xi_1,\ldots,\xi_p)=\sum_{s=1}^p\sum_{\substack{i_1,\ldots,i_{s-1}=-n+1\\i_{s+1},\ldots,i_p=-n+1}}^{n-1}c_{i_1\cdots i_{s-1}ni_{s+1}\cdots i_p}\big(\xi_1^{i_1}\cdots\xi_{s-1}^{i_{s-1}}\xi_{s}^n\xi_{s+1}^{i_{s+1}}\cdots\xi_p^{i_p}\\&&\qquad 
-\xi_p^{-i_1}\cdots\xi_{p-s+2}^{-i_{s-1}}\xi_{p-s+1}^{-n}\xi_{p-s}^{-i_{s+1}}\cdots\xi_1^{-i_p}\big)+\sum_{i_1,\ldots,i_p=-n+1}^{n-1}c_{i_1\cdots i_p}(\xi_1^{i_1}\cdots\xi_p^{i_p}-\xi_1^{-i_p}\cdots\xi_p^{-i_1}),
\end{eqnarray*}
where $p\ge 2$. In particular, for $p=2$,
\begin{eqnarray*}
&&a_2(\xi_1,\xi_2)=\sum_{i_{2}=-n+1}^{n-1}c_{ni_{2}}(\xi_1^{n}\xi_2^{i_2}-\xi_1^{-i_2}\xi_2^{-n})+\sum_{i_{1}=-n+1}^{n-1}c_{i_{1}n}(\xi_1^{i_1}\xi_2^{n}-\xi_1^{-n}\xi_2^{-i_1})\\
&&\qquad\qquad+\sum_{i_1,i_2=-n+1}^{n+1}c_{i_1i_2}(\xi_1^{i_1}\xi_2^{i_2}-\xi_1^{-i_2}\xi_2^{-i_1}).
\end{eqnarray*}
The classification process splits into two steps:

{\it Step 1}. The coefficients $\phi_{pq},\,p+q\le 3$, of the canonical formal recursion operator ($\phi(\eta)=\eta$) can be found explicitly using (\ref{frop1001})-(\ref{froppq}). One can show that the terms $\hu_l^p\hu_r^q\phi_{pq},\,p+q<3$ are quasi-local for any choice of constants $\alpha_k$ and $c_{i_1i_2},\,c_{i_1i_2i_3}$. The requirement of quasi-locality of terms $\hu_l^p\hu_r^q\phi_{pq},\,p+q=3$, results in a system of polynomial equations on $\alpha_k$ and $c_{i_1i_2},\,c_{i_1i_2i_3},\,c_{i_1i_2i_3i_4}$, which are necessary integrability conditions. The action of the re-scaling group
$$u_k\to\mu u_k,\quad t\to\nu t,\quad \mu,\nu\in\C^*$$
on constants $\alpha_k$ and $c_{i_1i_2},\,c_{i_1i_2i_3},\,c_{i_1i_2i_3i_4}$ is given by
$$
\alpha_k\to \nu^{-1}\alpha_k,\quad c_{i_1\cdots i_p}\to\mu^{p-1}\nu^{-1}c_{i_1\cdots i_p}.
$$
Modulo the action of this group, the set of solutions to the system of necessary integrability conditions is finite. Thus we obtain the list of $3$-approximately integrable equations of the form (\ref{eqclass}).

{\it Step 2}. From Proposition \ref{prounique} it follows the the requirement of quasi-locality of terms $\phi_{pq},\,p+q>3$ uniquely determines the terms $\hu^{p+q}a_{p+q}(\xi_1,\ldots,\xi_{p+q})$ and imposes further restrictions on constants $\alpha_k$ and $c_{i_1i_2},\,c_{i_1i_2i_3}, \,c_{i_1i_2i_3i_4}$. The obtained sequence of terms $\hu^la_l(\xi_1,\ldots,\xi_l)$ truncates at $l=13$. The process results in the lists of equations presented in Theorems \ref{class1}-\ref{class3}. $\hfill\square$

%{\it Step 3}. We prove integrability of each equation from the list either by transforming it to a known integrable equation or by presenting a Lax representation. $\square$
Before concluding this section, we compare our lists to the available ones for the abelian case. Similar to nonabelian PDEs, certain integrable equations have no integrable nonabelian lift, while others have two distinct lifts. When $n=1$, we get all three known nonabelian integrable D$\Delta$Es in our considered class as listed in Theorem \ref{class1}.
For $n=2$, the classification of quasi-linear equations of order $(-2,2)$ in the form
$$
u_t=A(u_{-1},u,u_1)u_2+B(u_{-1},u,u_1)u_{-2}+C(u_{-1},u,u_1),
$$
where $A$, $B$, and $C$ are functions of their variables, has recently been carried out in \cite{YamGL1,YamGL2}. This is a wider class of equations than the one we considered. Additionally, there is no requirement on skew-symmetricity.
The authors produced complete lists of integrable equations admitting symmetries of either orders $(-n,n)$ for all $n\in\mathbb{N}$ or only even-order symmetries $(-2n,2n)$, where $n\in\mathbb{N}$.
We present the correspondences between equations in Theorem \ref{class2} to their lists using the equation labels ${\rm (E.x)}$ and ${\rm (E.x')}$ used in \cite{Yam19}:
\begin{eqnarray*}
&&\mbox{\eqref{Volterra2}}\rightarrow \E.1; \quad
 \left. \begin{array}{c}\mbox{\eqref{mVolterra21}} \\ \mbox{\eqref{mVolterra22}}\end{array}\right\} (\alpha= 0) \rightarrow \E.2;\quad
\left. \begin{array}{c}\mbox{\eqref{mVolterra21}} \\ \mbox{\eqref{mVolterra22}}\end{array}\right\} (\alpha\neq 0) \rightarrow \E.3;\quad
\mbox{\eqref{eq45A}}\rightarrow \E.1';\\
 &&\mbox{\eqref{eq80A}} \rightarrow \E.4;\quad
 \left. \begin{array}{c}\mbox{\eqref{eq81}} \\ \mbox{\eqref{eq80B}}\end{array}\right\} (\alpha\neq 0) \rightarrow \E.6;\quad
 \left. \begin{array}{c}\mbox{\eqref{eq81}} \\ \mbox{\eqref{eq80B}}\end{array}\right\} (\alpha= 0) \rightarrow \E.2';
 \quad \mbox{\eqref{eq47}}\rightarrow \E.6';\\
 && \mbox{\eqref{eq79A}} \rightarrow \E.8';\quad
 \mbox{\eqref{eq74}}\rightarrow \E.7'(\alpha=0);\quad   \mbox{\eqref{eq45B}} \rightarrow \E.7'; \quad
 \mbox{\eqref{eq30}}\rightarrow \E.13'; \quad  \mbox{\eqref{eq100}} \rightarrow \E.11;\\
  &&\!\!\! \left. \begin{array}{l}\mbox{\eqref{eq136}} \\ \mbox{\eqref{eq132A}}\end{array}\right\} \rightarrow \E.9';\quad\ \
 \left. \begin{array}{c}\mbox{\eqref{eq132B}} \\ \mbox{\eqref{eq144}}\end{array}\right\} \rightarrow \E.13; \quad\ \
   \mbox{\eqref{eq147A}}\rightarrow \E.12';\quad\ \
    \left. \begin{array}{c}\mbox{\eqref{eq147}} \\ \mbox{\eqref{eq148}}\end{array}\right\} \rightarrow \E.11'.
\end{eqnarray*}
Equations $\E.15$ and $\E.14'$ have no integrable lifts although they belong to the class considered in Theorem \ref{class2}.
When $n=3$, the class of equations considered in Theorem \ref{class3} differs from the one classified in \cite{MNWDiff}. For instance, Equations \eqref{eq32AV3}-\eqref{eq597V3} and \eqref{BGV43}-\eqref{BGV83} fall outside the classification presented in \cite{MNWDiff}. Further details are omitted.

\section{Concluding remarks}

In this paper we have developed a new approach that enables the derivation of explicit and easily verifiable necessary integrability conditions for scalar evolutionary D$\Delta$Es on associative algebras. We prove that if an equation
$$
u_t=f(u_{-n},\ldots,u_n),\quad f=\sum_{p\ge 1}f_p,\quad f_p\in\A_p
$$
possesses an infinite dimensional algebra of symmetries, then there exists a canonical formal recursion operator
$$
\Lambda=\eta+\sum_{p+q\ge 1}\hu_l^p\hu_r^q\phi_{pq}(\xi_1,\ldots,\xi_p,\eta,\zeta_1,\ldots,\zeta_q),
$$
where the coefficients are explicitly expressed in terms of the symbolic representation of the equation. The requirement of quasi-locality of all terms $\hu_l^p\hu_r^q\phi_{pq}(\xi_1,\ldots,\xi_p,\eta,\zeta_1,\ldots,\zeta_q)$ provides explicit necessary integrability conditions for the equation. Moreover, these conditions are independent on the symmetry structure of the equation.

We applied our methodology to classify integrable quasi-linear skew-symmetric equations of the form (\ref{eqclass}) of orders $(-1,1)$, $(-2,2)$, and $(-3,3)$. For each equation on the list, we presented either its Lax representation or a Miura transformation to a known integrable equation. We conjecture that each equation on the list is a member of an infinite family of integrable equations of arbitrarily high order. For example, equations (\ref{eq144}) and (\ref{BGV43}) are members of a family of multiplicative nonabelian Bogoyavlensky equations of the form (\ref{BGmult2n}). The description of such families is outside the scope of this paper.

The requirement of quasi-linearity and skew-symmetry in the classification problem may be removed.
While the classification of non-quasi-linear equations remains unresolved, the approach developed in this paper is suitable to tackle this problem.
Indeed, the derived integrability conditions are applicable in the case of a generic equation. It is worth noting that all presently known integrable nonabelian non-quasi-linear equations are related to the quasi-linear ones through either invertible or non-invertible Miura-type transformations.
This observation prompts speculation that such a relationship may universally apply to integrable non-quasi-linear equations.
Additionally, all non-skew-symmetric equations known to the authors are either connected to linear ones or to non-linearisable skew-symmetric equations. The verification (or refutation) of these hypotheses remains a challenge for further investigations.

 Integrable D$\Delta$Es serve as generalized symmetries for integrable discrete equations and as B\"acklund transformations for integrable PDEs. Although extensively studied in the abelian case since the inception of integrable systems theory, investigations in the nonabelian setting remain notably limited. The construction of nonabelian integrable discrete systems corresponding to the D$\Delta$Es presented in this paper, in particular, for the new integrable equations \eqref{eq30} and \eqref{BGmult2n}, is of significant interest.

The nonabelian equations can also be interpreted as quantised versions of classical systems. For instance, the nonabelian KdV equation was derived by transitioning classical fields to a quantum framework and substituting Poisson brackets with commutators \cite{fc95}.
Recently, a new approach to the problem of quantisation based on the notion of quantisation ideals was proposed by Mikhailov in \cite{avm20}. This
method switches the focus from deformations of Poisson manifolds to the dynamical systems defined on quantised algebras, i.e. algebras satisfying commutation relations that are compatible with the dynamics. For a detailed exploration of this direction concerning the Volterra chain, one may refer to \cite{CMW, cmw}.
In this new approach it is proposed to start from a dynamical system defined on a free associative algebra. The nonabelian integrable D$\Delta$Es presented in this paper will serve as the ideal starting point for further exploring this quantisation framework.

\section*{Acknowledgements}

The authors are grateful to Alexander Mikhailov and Evgeny Ferapontov for very fruitful and stimulating discussions while working on this paper. This paper is partially supported by the EPSRC Small Grant Scheme EP/V050451/1.

\end{document}